\documentclass{article} 

\usepackage{amsmath}
\usepackage{amssymb}

\usepackage{latexsym} 
\usepackage{amsfonts}

\usepackage{amsthm}

\usepackage{graphicx}

\newcounter{NN}
\newcommand{\beq}{\begin{equation}}  
\newcommand{\eeq}{\end{equation}}   
\newcommand{\bear}{\begin{array}}  
\newcommand{\eear}{\end{array}}

\newtheorem{example}[NN]{Example}

\newtheorem{remar}[NN]{Remark} 
\newtheorem{propositio}[NN]{Proposition} 
\newtheorem{theore}[NN]{Theorem} 
\newtheorem{lemm}[NN]{Lemma} 

\newcommand\dd{\mathrm{d}}

\def\Z{\mathbb{Z}}

\def\Z{\mathbb{Z}}

\def\t#1{\widetilde{#1}}
\def\h#1{\widehat{#1}}

\title{Integrability of reductions of the  discrete KdV and potential KdV equations}
\author{A.N.W. Hone \\
School of Mathematics, Statistics \& Actuarial Science\\ 
University of Kent, Canterbury, U.K.
\and 
P.H. van der Kamp \& G.R.W. Quispel \\ 
Department of Mathematics and Statistics\\ 
La Trobe University, Melbourne, Australia.\\
\and   
D.T. Tran \\
School of Mathematics and Statistics\\
University of New South Wales, Sydney, Australia.
} 

\begin{document}
\maketitle





\abstract{We study the integrability of mappings obtained as reductions of the discrete   Korteweg-de Vries (KdV) equation
and of two copies of the discrete potential Korteweg-de Vries equation (pKdV). %
We show that 
the mappings corresponding to the discrete KdV equation, which can be derived from the 
latter,  are completely integrable in the  Liouville-Arnold sense. 
The mappings associated with two copies of the pKdV
equation are also shown to be 
integrable. 
} 


%
%

\section{Introduction} 

The problem of integrating differential equations goes back to the origins of calculus and 
its application to problems in classical mechanics.  In the nineteenth century, the notion of 
complete integrability was provided with a solid theoretical foundation by Liouville, whose 
theorem gave sufficient conditions for a Hamiltonian system to be integrated by quadratures; 
yet only a few examples of integrable mechanical systems (mostly with a small number of degrees of 
freedom) were known at the time.   
Poincar\'e's subsequent results on the non-integrability of the three-body problem seemed to indicate that 
many, if not most, systems 
should 
be non-integrable.  
Nevertheless, examples of integrable 
systems (and action-angle variables in particular) played an important role in the early development of quantum theory.


The theory of integrable systems only began to expand rapidly in the latter part of the twentieth 
century, with the discovery of the remarkable properties of the   Korteweg-de Vries (KdV) equation, together with 
a host of other 
nonlinear partial differential equations that were found to be amenable to the inverse scattering technique. 
As well as having exact pulse-like solutions (solitons) which undergo elastic collisions, such equations 
could be interpreted as infinite-dimensional Hamiltonian systems, with an infinite number of conserved quantities. 
Moreover, it was shown that these equations admit particular reductions (e.g. to stationary solutions, or 
to travelling waves) which can be viewed as integrable mechanical systems with finitely many degrees of freedom.  
The papers in the collection \cite{zakharov} provide a concise and self-contained survey of the theory of integrable ordinary 
and partial differential equations;  for a more recent set of review articles, see \cite{mik}. 

In the last two decades or so, there has been a more gradual development of 
{\it discrete} 
integrable systems, in the form of 
finite-dimensional maps \cite{qrt1, CapelNijhoff1991} and 
discrete Painlev\'e equations \cite{rg}, 
as well as partial difference equations 
defined on lattices or quad-graphs \cite{Nijhoff1995KdV, ABS}. Discrete integrable systems 
can be obtained directly by seeking discrete analogues of particular continuous soliton equations or 
Hamiltonian flows \cite{suris}, but they also appear independently in solvable 
models  of statistical mechanics (see 
the link with the hard hexagon model in \cite{qrt1}, for example)   
or quantum field theory (see \cite{zabrodin}, for instance). 
An important  theoretical result  for ordinary difference equations or maps 
is the fact  that the Liouville-Arnold definition of integrability for 
Hamiltonian systems of ordinary differential equations %
can be extended 
naturally to symplectic maps \cite{maeda, Bruschi1991, veselov}, so that an appropriate 
modification of Liouville's theorem holds.  For lattice equations with two or three independent variables 
there is 
less theory available (especially from the Hamiltonian point of view), and the full details 
of the known integrable examples are still being explored, but one way to gain understanding 
is through the analysis of particular families of reductions. 

By imposing a periodicity
condition, integrable lattice  equations can be reduced to ordinary difference equations (or mappings/maps) 
\cite{CapelNijhoff1991, QCPN1991, Kamp2009Initialvalue,  KampStaircase}.
It is believed that the reduced maps
obtained from an {\em  integrable} lattice equation are completely  integrable in the Liouville-Arnold sense. 
To {\em prove}  that a map is integrable  one needs to find
a Poisson structure together with  a sufficient number of functionally independent first integrals, 
and then show that these integrals
commute with respect to the Poisson bracket. 
One complication that immediately arises is that the reduced maps naturally come in families
of increasing dimension, and the number of first  integrals grows with the dimension. 
The complete integrability of some particular KdV-type maps was proved  
in \cite{CapelNijhoff1991}, and 
progress has been made recently with other families of maps.
For maps obtained as 
reductions of the equations in the Adler-Bobenko-Suris  (ABS) 
classification \cite{ABS}, 
and for reductions of the sine-Gordon and modified Korteweg-de Vries (mKdV) equations, first integrals 
were given in closed form by using the staircase method and the noncommutative Vieta expansion \cite{Tranclosedform, Kampclosedform}.
In particular, the complete integrability of mappings obtained as reductions of the discrete sine-Gordon, mKdV and potential KdV (pKdV)   
equations was studied in detail in \cite{TranPhD, TranInvo}.

Given a map, the question arises as to whether  
it has a Poisson structure, and if so,  how can one  find it?
In general, 
the  answer is not known. 
However, for some classes of maps one can assume that in coordinates $x_j$ 
the Poisson structure is in canonical or log-canonical form, 
i.e. 
the Poisson brackets  have the form $\{x_i,x_j\}=\Omega_{ij}$  or $\{x_i,x_j\}=\Omega_{ij}\, x_ix_j$ respectively,
where $\Omega$ is a constant skew-symmetric  matrix.
This approach  is effective 
for mappings obtained in the context of cluster algebras 
\cite{gsvduke, FordyHoneSIGMA, fordyhonecluster}, and also applies to  
reductions of the 
lattice pKdV, sine-Gordon and mKdV equations \cite{TranInvo}.
Another approach  requires the existence of a  Lagrangian for the reduced map: 
by using a discrete analogue of the Ostrogradsky transformation, as introduced in  \cite{Bruschi1991},
one can rewrite the map in  canonical coordinates; from there one can derive a Poisson structure in the original variables.

Here we start by considering a well known integrable lattice equation,  
namely the discrete potential Korteweg--de Vries 
 equation, 
also referred to as  $H_1$ in the  ABS list  \cite{ABS}, which is given by 
\begin{equation}
\label{E:pkdv} 
(u_{\ell , m} -u_{\ell +1,m+1})(u_{\ell +1,m}-u_{\ell ,m+1})=1, 
\end{equation}
where $(\ell , m)\in\mathbb{Z}^2$.  
Early results on this equation appear in  
\cite{qncl, wiersmacapel, CapelNijhoff1991} and \cite{Nijhoff1995KdV},   
where amongst other things it was shown that 
(\ref{E:pkdv}) leads to the 
continuous potential  KdV   equation, that is 
\beq\label{cpkdv} 
\frac{\partial u}{\partial t}  =\frac{\partial^3 u}{\partial x^3}+3\left(\frac{\partial u}{\partial x}\right)^2, 
\eeq
by  performing
a suitable continuum limit. In \cite{CapelNijhoff1991} 
a  
Lagrangian was 
obtained for the discrete pKdV equation (\ref{E:pkdv}); 
this can be explained 
using the so-called three-leg form, as in \cite{ABS}.  
However, the associated Euler-Lagrange equation turns out to 
consist of two copies of the lattice pKdV equation, 
that is 
\beq \label{double} 
 u_{\ell +1,m}+u_{\ell -1,m}+ \frac{1}{u_{\ell -1 ,m-1}-u_{\ell , m}} =
u_{\ell , m+1}+u_{\ell ,m-1} +\frac{1}{u_{\ell ,m}-u_{\ell +1 , m+1}} .  
\eeq  
The latter equation, which henceforth we refer to as 
the {\it double pKdV equation},  
is somewhat more general than (\ref{E:pkdv}): every solution 
of (\ref{E:pkdv}) is a solution of (\ref{double}), but the converse statement does not hold. 
In this paper we shall be concerned with the 
double pKdV equation 
(\ref{double}), rather than with  
(\ref{E:pkdv}). 

Next, we introduce a  new variable on the lattice,  $v_{\ell , m}:=u_{\ell,m}-u_{\ell+1,m+1}$,   
and immediately find that, whenever $u_{\ell , m}$ is a solution of (\ref{double}), 
$v_{\ell , m}$ satisfies 
\beq \label{E:kdv}
 v_{\ell +1 ,m}-v_{\ell ,m+1} = \frac{1}{v_{\ell ,m}}-\frac{1}{v_{\ell +1 ,m+1}} 
. 
\eeq 
The latter equation  is known as the lattice KdV equation. 
Both (\ref{E:pkdv}) and (\ref{E:kdv}) are 
integrable lattice 
equations, 
in the sense that they can be derived as the compatibility condition for 
an associated linear system, known as a Lax pair; this is discussed in section 4.

In this paper we perform the 
so-called $(d-1,-1)$-reduction 
on the 
discrete pKdV  
Lagrangian 
and derive the corresponding %
reduction of the double pKdV  
equation  (\ref{double}).  This means that we consider functions 
$u=u_{\ell , m}$ on the lattice   
which have the following periodicity property under shifts: 
\beq \label{N1} 
u_{\ell , m} =u_{\ell + d-1 , m-1}. 
\eeq 
Such periodicity implies that $u$ depends on the lattice variables $\ell$ and $m$ through the combination 
$n=\ell+m(d-1)$ only; thus, with a slight abuse of notation, we write $u=u_n$. 
Such a reduction can be understood as a discrete analogue of the travelling wave reduction of a 
partial differential equation: for a function $u(x,t)$ satisfying a suitable partial differential equation such as (\ref{cpkdv}), 
one considers solutions that are invariant under $x\to x+c\delta $, $t\to t+\delta$ for all $\delta$; such 
solutions depend on $x,t$ through the combination $z=x-ct$ only, corresponding to waves 
travelling with constant speed $c$. 
By analogy with the continuous case, where one obtains ordinary differential equations (with independent 
variable $z$) for the travelling waves, it is apparent that imposing the  condition (\ref{N1})  in the discrete setting
leads to ordinary difference equations (with independent 
variable $n$). 
Note that in the continuous case this reduction yields 
a one-parameter family of ordinary differential equations of fixed order  
(with parameter $c$), 
whereas in the discrete case one finds 
a family of ordinary difference equations 
whose order depends on $d$.    

Here we are concerned with the  complete integrability
of the ordinary difference equation obtained as the $(d-1,-1)$-reduction of 
the double pKdV equation (\ref{double}), 
which is equivalent to a birational map in dimension $2d$, 
and the associated reduction of the 
lattice KdV equation (\ref{E:kdv}), which gives a map in dimension $d$. 
We begin by finding Poisson brackets 
for  reductions of the double pKdV equation, 
which are then used to infer brackets for the corresponding maps 
obtained from lattice KdV; having found first integrals and 
proved Liouville integrability for the KdV maps, we are subsequently able 
to do the same for the double pKdV maps.

The paper is organized as follows. In section 2, we start with 
the discrete Lagrangian whose 
Euler-Lagrange equation is the double pKdV 
equation (\ref{double}). 
For each $d$, we then derive a symplectic structure for the 
double pKdV map 
obtained as the 
$(d-1,-1)$-reduction of this discrete Euler-Lagrange equation, by using 
a discrete analogue of the Ostrogradsky
transformation. 
This provides a nondegenerate Poisson bracket for each of these maps. 
In section 3 we present a  Poisson structure for the 
associated map obtained 
as the $(d-1,-1)$-reduction of the lattice 
KdV equation (\ref{E:kdv}), 
which is induced from the bracket found in section 2. 
The  Hirota bilinear form of each of the KdV maps is also given, 
in terms of tau-functions, 
whence (via a link with cluster algebras) we derive a second Poisson structure that is compatible with the 
first. In section 4 we present closed-form expressions for 
first integrals of 
each reduced KdV map, 
and we show that they are in involution with respect 
to the first of the Poisson structures. 
This furnishes a direct proof of 
Liouville integrability for these KdV maps, 
within the framework of the papers \cite{Tranclosedform} 
and \cite{TranInvo}. 
Another proof,  
based on the pair of compatible Poisson brackets, is also sketched.  
In section 5 we return to the double pKdV maps, 
and present first integrals for each of them.  
Some of these integrals are derived from the integrals of 
the corresponding KdV map in the previous section, 
while additional commuting integrals 
are found using a function  periodic with period $d$; 
this is a $d$-integral, in the sense of \cite{haggar}. The paper ends with some brief conclusions, 
and some additional comments are relegated to an Appendix. 


\section{
Poisson 
structure 
for the double pKdV maps}

Henceforth, it is convenient to use $\t{ \ }$   and  $\h{\ }$ to denote  shifts in $\ell$ and $m$ directions respectively, 
so that 
$\widetilde{u}=u_{\ell+1,m}$, $\widehat{u}=u_{\ell,m+1}$, etc. 
It is known that 
all equations  of the form 
\beq\label{abs} 
Q(u,\t{u},\h{u},\h{\t{u}},\alpha,\beta)=0 
\eeq 
in the ABS list \cite{ABS} 
are equivalent (up to some transformations) to the existence of an equation of so-called three-leg type, 
that is  
\begin{equation}
\label{E:3leg_form}
P(u,\t{u},\h{u},\h{\t{u}};\alpha,\beta)\equiv\phi(u,\h{u},\alpha)-\phi(u,\t{u}, \beta)-\psi(u,\h{\t{u}},\alpha,\beta)=0,
\end{equation}
for suitable functions $\phi, \psi$, where $\alpha$ and $\beta$ are parameters; the latter 
leads to the derivation of a Lagrangian for each of the equations (\ref{abs}). 
In particular, for the pKdV equation (\ref{E:pkdv}), which is a (parameter-free) equation 
of the form (\ref{abs}),  
we have $\phi(u,\t{u})=u+\t{u}$ and $\psi(u,\h{\t{u}})=1/(u-\h{\t{u}})$, 
and using the three-leg form (\ref{E:3leg_form}) leads to the Lagrangian
\begin{equation}
\label{E:Lagr}
\mathcal{L}=\frac{1}{2}(u+\t{u})^2-\frac{1}{2}(u+\h{u})^2-\log |u-\h{\t{u}}|.
\end{equation}
The corresponding discrete Euler-Lagrange equation  is the 
double pKdV 
equation (\ref{double}), which can be rewritten as 
\beq\label{doublej}
\h{\widetilde{\mathrm{J}}}
-
\mathrm{J} 
=0 
,  
\qquad \mathrm{with} \quad  
\mathrm{J} =  
\t{u} -\h{u}-\frac{1}{ (u- \h{\t{u}}) } .  
\eeq 
The latter equation is more general than (\ref{E:pkdv}), which arises in the special case that J 
is identically zero\footnote{However, 
every solution of (\ref{double}) can be written as $u=U+a$, where 
$U$ is a solution of  (\ref{E:pkdv}) and $a$ is a solution of the 
linear equation $\widehat{\widetilde{a}}=a$.  
See the Appendix for more details. 
}.

Now 
setting $n=\ell+m(d-1)$, with $u$ satisfying (\ref{N1}), 
the $(d-1,-1)$-reduction applied to the Lagrangian 
in (\ref{E:Lagr}) gives  $\mathcal{L} = \mathcal {L}\left( u_n,u_{n+1},\ldots,u_{n+d}\right)$, where   
\begin{equation}
\label{E:ReLagr}
\mathcal{L}
=\frac{1}{2}(u_n+u_{n+1})^2-\frac{1}{2}(u_n+u_{n+d-1})^2-\log |u_n-u_{n+d}|.
\end{equation}
The discrete action functional is
$ \mathrm{S}:=\sum_{n\in\Z} \mathcal{L}\left(u_n,u_{n+1},\ldots,u_{n+d}\right)$.
It yields the discrete Euler-Lagrange equation 
\begin{equation}
\label{E:Eu-LagrEq.}
\frac{\delta \mathrm{S}}{\delta u_n}=\sum_{r=0}^{d}\frac{\partial \mathcal{L}\left(u_{n-r},u_{n+1-r},\ldots,u_{n+d-r}\right)}{\partial u_n}
=\sum_{r=0}^{d} \mathcal{E}^{-r}\mathcal{L}_r=0,
\end{equation}
where $\mathcal{E}$ denotes the shift operator and
$\mathcal{L}_r=\frac{\partial \mathcal{L}\left(u_n,u_{n+1},\ldots,u_{n+d}\right)}{\partial u_{n+r}}$.
Thus 
we obtain the ordinary difference equation
\begin{equation}
\label{E:2copies}
u_{n+1}-u_{n+d-1}+u_{n-1}-u_{n-d+1}-\frac{1}{u_n-u_{n+d}}+\frac{1}{u_{n-d}-u_n}=0, 
\end{equation}
which is precisely the $(d-1,-1)$-reduction of (\ref{double}). 
The solutions of this equation are equivalent to the iterates of the $2d$-dimensional map
\begin{equation}
\label{E:pKdVmap}
(u_{n-d},u_{n-d+1},\ldots, u_{n+d-1})\mapsto (u_{n-d+1},u_{n-d+2},\ldots, u_{n+d}),
\end{equation}
where 
$u_{n+d}$ 
is found from equation~\eqref{E:2copies}. 

Given a Lagrangian of first order for a classical mechanical system, the Legendre transformation produces canonical symplectic coordinates on the phase space; 
the Ostrogradsky transformation is the analogue of this for Lagrangians of higher order \cite{blaszak}. 
In order to derive a nondegenerate Poisson bracket for the $2d$-dimensional map, 
we  use a discrete analogue of the Ostrogradsky transformation, as  given in \cite{Bruschi1991}, 
which is a change of variables to canonical coordinates,   
$\left(u_{n-d},u_{n-d+1}, \ldots, u_{n+d-1}\right) 
\to 
 (q_1,\ldots, q_d,p_1,\ldots,p_d)$, 
 where
$
q_i =u_{n+i-1}
$, 
$
p_i =\mathcal{E}^{-1}\sum_{r=0}^{d-i}\mathcal{E}^{-r}\mathcal{L}_{r+i}$.  
Thus, from (\ref{E:ReLagr}), we obtain
\begin{align*} 
q_i &=u_{n+i-1}, \quad  i=1, \ldots ,  d; \\
p_1&=-u_{n+1}+u_{n+d-1}+\frac{1}{u_n-u_{n+d}}=u_{n-1}-u_{n+1-d}+\frac{1}{u_{n-d}-u_n};\\
p_i&=-u_{n+i-d}-u_{n+i-1}+\frac{1}{u_{n-d+i-1}-u_{n+i-1}},\quad  i=2,\ldots ,d-1; \\
p_d&=\frac{1}{u_{n-1}-u_{n+d-1}}.
\end{align*}
In terms of the canonical coordinates $(q_j,p_j)$ the map (\ref{E:2copies}) is rewritten as 
$$ \bear{lcllcl} 
q_i & \mapsto & q_{i+1},\quad i=1,\ldots, d-1; \qquad 
& 
q_d & \mapsto & q_1- 
(p_1+q_2-q_d)^{-1}; \\
p_1 & \mapsto & p_2+q_2+q_1; \qquad 
& 
p_i & \mapsto & p_{i+1}, \quad i=2,\ldots, d-2;  \\ 
p_{d-1} & \mapsto & p_{d}-q_1-q_{d} ; \qquad  
&
p_d & \mapsto & p_1+q_2-q_d.
\eear 
$$ 
By the general results in \cite{Bruschi1991}, this map is symplectic with respect to the canonical symplectic form 
$
\sum_{j=1}^d \dd p_j\wedge \dd q_j$. Equivalently, it preserves the canonical Poisson brackets 
$\{p_i,p_j\}=\{q_i,q_j\}=0$,
$\{p_i,q_j\}=\delta_{ij}$. 

In order to find the Poisson brackets for the coordinates $u_n$, we write them 
in terms of $(q_j,p_j)$.
For all $0\leq i\leq d-1$, we have $u_{n+i}=q_{i+1}$, and
$$ 
\bear{rcl}
u_{n-1} &=& q_d+1/p_d=[q_d; p_d],\\
u_{n-i-1} & = &[q_{d-i};p_{d-i}+q_{d-i}+u_{n-i}] \qquad (\mathrm{for}\,  \, 1\leq i\leq d-2 ), \\
u_{n-d}& = &[q_1; p_1+u_{n+1-d}-u_{n-1}] , \\
\eear
$$
which means that for $1\leq i\leq d-2$, $u_{n-i-1}$ is found recursively as 
$$ 
[q_{d-i}; p_{d-i}+q_{d-i}+q_{d-i+1},p_{d-i+1}+q_{d-i+1}+q_{d-i+2},\ldots, p_{d-1}+q_{d-1}+q_d, p_d], 
$$ 
while 
$u_{n-d}=[q_1; p_1-q_d+q_2-\frac{1}{p_d},p_2+q_2+q_3,\ldots,p_{d-1}+q_{d-1}+q_d, p_d ]$,
where $[\,\, ;\, ]$ denotes a continued fraction. This yields the following result. 
\begin{theore}\label{ubracket}
The $2d$-dimensional map  given by (\ref{E:pKdVmap}) with (\ref{E:2copies}) is a Poisson map with respect to the 
nondegenerate bracket given for $0\leq j<d-1$ by 
\beq
\bear{rcl}  \label{E:Poi1}
\{u_{n-d},u_{n-j-1}\} & = & 0,    
\\
\{u_{n-d},u_{n+j}\} & = & (-1)^{j+1}(u_{n-d}-u_n)^2\ldots (u_{n-d+j}-u_{n+j})^2, 
\\
\{u_{n-d},u_{n+d-1}\}& =& (-1)^d(u_{n-d}-u_n)^2\ldots (u_{n-1}-u_{n+d-1})^2 \\  
&& -(u_{n-d}-u_n)^2(u_{n-1}-u_{n+d-1})^2. 
\eear 
\eeq 
\end{theore}
\begin{proof} 
In order to prove  the first of the formulae in 
(2.9),  
we note that 
$$
\begin{array}{rl} 
\{ 
u_{n-d},u_{n-r-1} 
\} 
 = &
\{ 
[q_1; p_1-q_d+q_2-\frac{1}{p_d},p_2+q_2+q_3,\ldots ,p_{d-1}+q_{d-1}+q_d, p_d ],  \\   

& 
[
q_{d-r}; p_{d-r}+q_{d-r}+q_{d-r+1}, 
\ldots 
, 
p_{d-1}+q_{d-1}+q_d, p_d 
] \}  
\\ 
  = &  
\sum_{i\geq d-r} \left(\frac{\partial u_{n-d}}{\partial p_i}\frac{\partial u_{n-r-1}}{\partial q_i}- 
\frac{\partial u_{n-d}}{\partial q_i}\frac{\partial u_{n-r-1}}{\partial p_i}
\right) .
\end{array} 
$$ 
Then we 
anticipate the next section by setting 
$v_{n}:=u_{n-d}-u_{n}$, 
to find that 
$$
\begin{array}{rcll} 
\frac{\partial u_{n-d}}{\partial p_i} & = &   
(-1)^i v_n^2\ldots v_{n+i-1}^2,  
& i<d ,  
\\
\frac{\partial u_{n-d}}{\partial p_d} & = &  
(-1)^d \, v_n^2\ldots v_{n+d-1}^2-v_n^2 v_{n+d-1}^2 , 
& \\
\frac{\partial u_{n-d}}{\partial q_i} & = & 
(-1)^{i-1}\, v_n^2\ldots v_{n+i-2}^2\big(1-v_{n+i-1}^2\big), 
& i<d , 
\\
\frac{\partial u_{n-d}}{\partial q_d} & = &  
v_n^2+(-1)^{d-1}\, v_n^2\ldots v_{n+d-2}^2, 
& \\
\frac{\partial u_{n-r-1}}{\partial p_i} & = &  
(-1)^{r+i-1-d}\, v_{n+d-r-1}^2\ldots v_{n+i-1}^2 , 
& \\
\frac{\partial u_{n-r-1}}{\partial q_i} & = &   
(-1)^{r+i-d}\, v_{n+d-r-1}^2\ldots v_{n+i-2}^2\big(1-v_{n+i-1}^2\big), 
 & i<d, \\
\frac{\partial u_{n-r-1}}{\partial q_d} & = & 
(-1)^{r}\, v_{n+d-r-1}^2 \ldots v_{n+d-2}^2, 
\end{array} 
$$
for $i\geq d-r\geq 2$,    
and 
hence we obtain $\{u_{n-d}, u_{n-r-1}\}=0$. 
Next,  we use
$\{u_{n-d}, u_{n+j}\} 
=\{u_{n-d}, q_{j+1}\}=\frac{\partial u_{n-d}}{\partial p_{j+1}}$,
where $0\leq j\leq d-1$, from which the second and third formulae in (2.9) 
follow. 
\end{proof}

\section{Poisson structures and tau-functions for the KdV maps}
As mentioned above, the discrete KdV equation can be  derived from  
the 
double 
pKdV equation. 
This suggests that the symplectic structure given in the previous section 
can be used to find  a
Poisson structure for the 
$(d-1,-1)$-reduction of the discrete KdV equation. 
It turns out that, in addition to the bracket induced from 
double pKdV, 
each of the KdV maps has a second, independent Poisson bracket, which is obtained from 
a Hirota bilinear form in terms of tau-functions. The second bracket is constructed by 
making use of a connection with Somos recurrences 
and cluster algebras.  

\subsection{First Poisson structure from pKdV} 
Whenever $u_n$ is a solution of equation~\eqref{E:2copies}, 
$v_n=u_{n-d}-u_{n}$  
satisfies a difference equation of order $d$, namely 
 \begin{equation}
 \label{E:KdV_Reduction}
 v_{n+d-1}-v_{n+1}-\frac{1}{v_{n+d}}+\frac{1}{v_n}=0.
 \end{equation}
Alternatively, by starting from a solution of (\ref{double}) with the periodicity property (\ref{N1}), we 
see that this yields a solution $v=v_{\ell,m}$  of (\ref{E:kdv}) with the same periodicity, 
and so (writing this as $v_n$, with the same abuse of notation as before) 
it is clear that equation (\ref{E:KdV_Reduction}) is just 
the $(d-1,-1)$-reduction of the discrete KdV equation. 
Equivalently, the ordinary difference equation~\eqref{E:KdV_Reduction} corresponds to 
the $d$-dimensional  map 
\begin{equation}
\label{E:kdvmap} 
\varphi: \quad   
(v_0,v_{1},\ldots,v_{d-1})\mapsto\left(v_{1},v_{2},\ldots, v_{d-1},\frac{v_0}{1+v_{d-1}v_0-v_{1}v_0}\right). 
\end{equation}
The case $d=2$ is  trivial, so henceforth we consider $d\geq 3$. 

In the above, the suffix $n$ has been dropped, taking a fixed $d$-tuple $(v_0,v_1, \ldots ,v_{d-1})$ 
in $d$-dimensional  space. However, because the map (\ref{E:kdvmap}) is obtained from a 
recurrence of order $d$,  all of the formulae are invariant under simultaneous shifts of all indices 
by an arbitrary amount $n$, 
i.e. $v_j\to v_{n+j}$ for each $j$. For a fixed $n$, say $n=0$, the 
formulae in Theorem \ref{ubracket} define a Poisson bracket in dimension $2d$, which 
can be used to calculate the brackets between the quantities $v_j = u_{j-d}-u_j$ for $j=0,\ldots, d-1$.  
Remarkably, these brackets can be rewitten in terms of $v_j$ alone; in other words, these 
quantities form a Poisson subalgebra of dimension $d$. Hence   
this provides the  first of two ways to endow \eqref{E:kdvmap} with a Poisson structure. 
\begin{theore}\label{vbracket1}
The $d$-dimensional map (\ref{E:kdvmap})  preserves the 
Poisson 
bracket $\{ \, , \, \}_1$ defined by 
\begin{equation}
\label{E:KdVsymp}
\displaystyle{
\{v_{i},v_{j}\}_1=\left\{
\begin{array}{ll}
(-1)^{j-i}\prod_{r=i}^jv_{r}^2, & 0<j-i<d-1,\\
\left(1+(-1)^{d-1}\prod_{r=1}^{d-2}v_{r}^2\right) v_0^2v_{d-1}^2, & j-i=d-1.
\end{array}
\right. }
\end{equation} 
This bracket is degenerate, with one Casimir when $d$ is odd, and two independent Casimirs 
when $d$ is even. 
\end{theore}
The above result follows immediately from Theorem \ref{ubracket}, apart from the statement about 
the Casimirs, which will be explained shortly. Here we first give a couple of examples for illustration. 

\begin{example}\label{d3br1} 
{\em When $d=3$ the map  $\varphi$ given by (\ref{E:kdvmap}) preserves the 
bracket 
$$ 
\{ v_0,v_1\}_1 = -v_0^2v_1^2, \quad 
\{ v_1,v_2\}_1 = -v_1^2v_2^2, \quad 
\{ v_0,v_2\}_1 = (1+v_1^2)v_0^2v_2^2. 
$$
This 
Poisson  
bracket has rank two, with a Casimir $\cal C$ that is also a first integral for the map, i.e. 
$\varphi^* \cal C =\cal C$ with 
$$ 
{\cal C} = v_1-\frac{1}{v_0} -\frac{1}{v_1}-\frac{1}{v_2}. 
$$  
}
\end{example} 

\begin{example}\label{d4br1} 
{\em 
When $d=4$ the map  (\ref{E:kdvmap}) preserves the 
bracket specified by 
$$ 
\{ v_0,v_1\}_1 = -v_0^2v_1^2, \quad 
\{ v_0,v_2\}_1 = v_0^2v_1^2v_2^2, \quad 
\{ v_0,v_3\}_1 = (1-v_1^2v_2^2)v_0^2v_3^2, 
$$ 
where all other brackets $\{v_i,v_j\}_1$ for $0\leq i,j\leq d-1$ are determined from 
skew-symmetry and the Poisson property of $\varphi$. 
This is a bracket of rank two, having two independent Casimirs given by 
$$ 
{\cal C}_1= v_1-\frac{1}{v_0} -\frac{1}{v_2}, \quad {\cal C}_2= v_2-\frac{1}{v_1} -\frac{1}{v_3}, 
\qquad 
\mathrm{with} \quad \varphi^*{ \cal C}_1 ={\cal C}_2, \quad \varphi^* {\cal C}_2 ={\cal C}_1. 
$$  
}
\end{example} 

\begin{remar} 
The Casimirs ${\cal C}_j$ in the preceding example are 2-integrals \cite{haggar}, meaning that 
they are preserved by two iterations 
of the map, i.e. $(\varphi^*)^2 {\cal C}_j =  {\cal C}_j$ for $j=1,2$. The symmetric functions 
\beq\label{K} 
{\cal K} = {\cal C}_1+{\cal C}_2, 
\qquad {\cal K}' = {\cal C}_1{\cal C}_2 
\eeq 
provide two independent first integrals.  
\end{remar}

In order to make the properties of the Poisson bracket $\{ \, ,\,\}_1$ more  
transparent, we introduce some new coordinates, the motivation for which should 
become clear from 
the Lax pairs in section 4.  

\begin{lemm}
\label{gcoords} 
For all $d\geq 4$, 
with respect to the coordinates 
\beq\label{gs} 
g_0=-1/v_0, \qquad 
g_j =v_{j-1}-1/v_j, \quad j=1,\ldots , d-1, 
\eeq 
the first Poisson bracket for the KdV map (\ref{E:kdvmap})  is specified by the 
following relations for $\nu =1$ and $0\leq i<j\leq d-1$:
\begin{equation}
\label{C6E:KdVPoi}
\{g_{i},g_{j}\}_1=\left\{
\begin{array}{ll}
-1& \mbox{if}\  \, j-i=1,\\
1 & \mbox{if}\  \, j-i=d-1,\\
\nu /g_{0}^2& \mbox{if}\  \,i=1, j=d-1,\\\
0&\mbox{otherwise}.
\end{array}
\right.
\end{equation}
When $d$ is odd this bracket has the Casimir 
\beq\label{oddc} 
{\cal C} = g_0 + g_1 + \ldots + g_{d-1} + \frac{\nu}{g_0},   
\eeq  
while for even $d$ there are the two Casimirs 
\beq\label{evenc} 
{\cal C}_1 = g_0 + g_2+ \ldots + g_{d-2}, 
\qquad 
{\cal C}_2 =g_1 + g_3+ \ldots + g_{d-1} + \frac{\nu}{g_0}.   
\eeq  
\end{lemm}

\begin{remar} 
When $\nu =0$ the Poisson bracket (\ref{C6E:KdVPoi}) is just the first Poisson bracket for the dressing chain, 
as given by equation (13) in \cite{shabves}. 
\end{remar}  

In terms of the coordinates $g_j$, the map (\ref{E:kdvmap})  is rewritten as  
\begin{equation}
\label{gkdv} 
\varphi: \quad   
(g_0,g_{1},\ldots,g_{d-1})\mapsto\left(g_{1}+\frac{1}{g_0},g_{2},\ldots, g_{d-1},\frac{g_0^2g_1}{1+g_0g_{1}}\right). 
\end{equation}
Note that for the special case $d=3$, as in Example \ref{d3br1}, part of the formula for the bracket (\ref{C6E:KdVPoi}) 
requires a slight modification, namely 
$ \{ g_1,g_2\}_1 = -1 + 1/g_0^2$.

\subsection{Second Poisson structure from cluster algebras for tau-functions} 

The discrete KdV equation (\ref{E:kdv}) was derived by Hirota in terms of tau-functions, 
via the B\"acklund transformation for the differential-difference KdV equation \cite{hirota}. 
In Hirota's approach, the solution of the discrete KdV 
equation is given in terms of a tau-function as $v= \t{ \tau } \h{\tau } / (\tau \h{ \t{\tau }})$, and at the 
level of the 
$(d-1,-1)$-reduction this becomes 
\beq\label{tauf} 
v_n = \frac{\tau_{n+1}\, \tau_{n+d-1}}{\tau_n\, \tau_{n+d}}. 
\eeq 
By direct substitution, it then follows that $v_n$ is a solution of (\ref{E:KdV_Reduction}) provided that 
$\tau_n$ satisfies the trilinear (degree three) recurrence relation 
\beq\label{tril} 
\bear{rcl} 
\tau_{n+2d}\, \tau_{n+d-1}\,\tau_{n+1}&  = & \tau_{n+2d-1}\, \tau_{n+d+1}\,\tau_{n} 
- \tau_{n+2d-1}\, \tau_{n+d-1}\,\tau_{n+2} \\ && + \tau_{n+2d-2}\, \tau_{n+d+1}\,\tau_{n+1}. 
\eear 
\eeq 
However, this relation can be further simplified upon dividing by $\tau_{n+1}\tau_{n+d}$, which 
gives the relation 
$$
\alpha_n = \frac{\tau_{n}\, \tau_{n+d+1} - \tau_{n+2}\tau_{n+d-1}}{\tau_{n+1}\,\tau_{n+d}} 
=  \frac{\tau_{n+d-1}\, \tau_{n+2d} - \tau_{n+d+1}\tau_{n+2d-2}}{\tau_{n+d}\,\tau_{n+2d-1}}=\alpha_{n+d-1}. 
$$ 
This immediately yields relations that are bilinear (degree two) in $\tau_n$. 
\begin{propositio}\label{somos} 
The solutions of the equation  (\ref{E:KdV_Reduction}) are given in terms of a tau-function 
by (\ref{tauf}), where $\tau_n$ satisfies the bilinear recurrence relation  
\beq\label{bil} 
\tau_{n+d+1}\, \tau_n =\alpha_n \, \tau_{n+d}\, \tau_{n+1} + \tau_{n+d-1}\, \tau_{n+2}
\eeq 
of order $d+1$, with the coefficient $\alpha_n$ having period $d-1$.      
\end{propositio}

Apart from the presence of the periodic coefficient $\alpha_n$, the bilinear relation (\ref{bil}) has the form 
of a Somos-$(d+1)$ recurrence \cite{somos}. Such recurrence relations (with constant coefficients) are 
also referred to as three-term Gale-Robinson recurrences  
(after \cite{gale} and \cite{robinson} respectively). 

\begin{example}\label{s4int} {\em For $d=3$ the equation (\ref{bil}) is a Somos-4 recurrence 
with coefficients of period 2, that is 
$
\tau_{n+4}\, \tau_n =\alpha_n \, \tau_{n+3}\, \tau_{n+1} +  \tau_{n+2}^2$, with  
$\alpha_{n+2}=\alpha_n$.  
Due to the Laurent phenomenon \cite{fzlaurent}, the iterates of this recurrence are Laurent 
polynomials, i.e.  
polynomials in the initial values data and their reciprocals with integer coefficients; to be precise, 
$ 
\tau_n \in \mathbb{Z} [\alpha_0 ,\alpha_1, \tau_0^{\pm 1},  \tau_1^{\pm 1},\tau_2^{\pm 1},\tau_3^{\pm 1}] 
$ 
for all $n\in\mathbb{Z}$. 
This means that 
integer sequences can be generated from a suitable choice of initial data and coefficients. For instance, 
with the initial values $\tau_0= \tau_1= \tau_2= \tau_3= 1$ and parameters $\alpha_0=1$, $\alpha_1=2$, 
the Somos-4 recurrence yields an integer sequence beginning with 
$1,1,1,1,2,5,9,61,193, 1439, 13892,121853, 1908467,47166783   , \ldots$.  
}
\end{example} 

From the work of Fordy and Marsh \cite{fordymarsh}, it is known 
that, at least in the case where the coefficients are constant, recurrences of Somos type can be 
generated from sequences of mutations in a cluster algebra. For the purposes of this paper, the main 
advantage of considering the cluster algebra is that it provides a natural presymplectic structure for the tau-functions.  
A presymplectic form that is compatible with cluster mutations was presented in \cite{gsvduke}, and in 
\cite{fordyhonecluster} 
it was explained how this presymplectic structure is preserved by the recurrences considered in \cite{fordymarsh}. 

Cluster algebras are a new class of commutative algebras which were introduced in \cite{fz1}. Rather than 
having a set of generators and relations that are given from the start, the generators of a cluster algebra 
are defined recursively by an iterative process known as cluster mutation. For a coefficient-free 
cluster algebra, one starts from an initial set of 
generators (the initial cluster) of fixed size, 
which here we take to be $d+1$. 
If the initial cluster is  denoted by $(\tau_1,\ldots ,\tau_{d+1})$, 
then for each index $k$ one defines the mutation in the $k$ direction to be the transformation 
that exchanges one of the variables to produce a new cluster     
$(\tau_1',\ldots ,\tau_{d+1}')$ given by  
\beq\label{exrel}
\tau_j'=\tau_j, \quad j\neq k, \qquad 
\tau_k'\, \tau_k = \prod_{j=1}^{d+1}\tau_j^{[b_{jk}]_+} + \prod_{j=1}^{d+1}\tau_j^{[-b_{jk}]_+}, 
\eeq 
where the exponents in the exchange relation for $\tau_k'$ come from an integer 
matrix $B=(b_{ij})$ known as the {\it exchange matrix}, and we have used the notation  
$[b]_+=\max (b,0)$. 

As well as cluster mutation, there is an associated operation of matrix 
mutation, which acts on the matrix $B$; the details of this are omitted here. Fordy and Marsh gave conditions 
under which skew-symmetric exchange matrices $B$ have a cyclic symmetry (or periodicity) under mutation, and 
classified all such $B$ with period 1 \cite{fordymarsh}. They also showed how  
this led to recurrence relations for cluster variables, by taking cyclic sequences of mutations. The requirement of periodicity puts conditions on the elements of the skew-symmetric matrix $B$, 
which (for a suitable 
labelling of indices) can be written as 
\beq\label{reln1}
b_{i,d+1}=b_{1,i+1}, \qquad i=1,\ldots , d,
\eeq
and
\beq\label{reln2}
b_{i+1,j+1}=b_{i,j}+b_{1,i+1} [-b_{1,j+1}]_+  - b_{1,j+1} [-b_{1,i+1}]_+  ,
\qquad 1\leq i,j \leq d. 
\eeq
The corresponding recurrence is defined by iteration of the map 
$
(\tau_1,  \ldots , \tau_d, \tau_{d+1}) \mapsto 
 (\tau_2,  \ldots , \tau_{d+1}, \tau_{1}')  
$
associated with the exchange relation (\ref{exrel}) for index $k=1$, 
where the exponents are given by the entries $b_{1,j}$ in the first row of $B$. Moreover, given a recurrence 
relation of this type, the conditions (\ref{reln1}) and (\ref{reln2}) allow 
the rest of the matrix $B$ to be constructed from the exponents corresponding to the first row, and 
these conditions are also necessary and sufficient for a  log-canonical 
presymplectic form $\omega$, as in (\ref{pre}) below, to be preserved (see Lemma 2.3 in \cite{fordyhonecluster}). 
In general this two-form is closed, but it may be degenerate.  

For the case at hand, the exponents appearing in the two monomials on the right hand side of (\ref{bil}) 
specify the first row of the matrix $B$ as 
$
(0, 1, -1, 0, \ldots , 0, -1, 1), 
$
and the rest of this matrix is found by applying  (\ref{reln1}) and (\ref{reln2}). 
Although the foregoing discussion was put in the context of coefficient-free cluster algebras, the presence of coefficients 
in front of these two monomials does not affect  the behaviour of the corresponding log-canonical two-form 
under iteration. Thus we obtain 
\begin{lemm}
\label{exchange} 
For all $d\geq 5$, the Somos-$(d+1)$ recurrence (\ref{bil}) 
preserves the presymplectic form 
\beq\label{pre} 
\omega = \sum_{i<j} \frac{b_{ij}}{\tau_i\tau_j} \, \dd\tau_i\wedge \dd\tau_j, 
\eeq 
given in terms of  the entries of the associated skew-symmetric exchange matrix  
$B=(b_{ij})$ of size $d+1$, where  (with an asterisk denoting the omitted entries below the diagonal) 
$$ 
B = \left( \begin{array}{ccrrccrr} 
0 & 1 & -1 & 0 & \cdots 
& 0 & -1 & 1 \\ 
   &  \ddots & 2 & -1 & 0 & \cdots 
& 1 & -1 \\ 
   &     & \ddots & \ddots & \ddots & \ddots 
& 
& 0 \\ 
  &    &          & \ddots & \ddots  & \ddots 
& \ddots  & \vdots 
\\  
 &     & &    & \ddots 
& \ddots & \ddots & 0 \\ 
&   &     & &    
& \ddots  & 2 & -1 \\  
&* &  & & 
& & 0 & 1 \\ 
& & & & 
& & & 0 
\end{array} \right) .
$$   
\end{lemm} 
\begin{remar}\label{s45} 
In each of the special cases $d=3,4$, the exchange matrix  does not fit into the general pattern above. 
For the details of the case $d=3$, when (\ref{bil}) is a Somos-4 relation, see Example 2.11 in \cite{fordyhonecluster}; 
and for $d=4$ (Somos-5) see Example 1.1 in \cite{FordyHoneSIGMA}. 
\end{remar} 
The exchange matrices being considered here are all degenerate: for $d$ odd, $B$ has a two-dimensional 
kernel, while for $d$ even, the kernel is three-dimensional. In the odd case, the kernel is associated with  
the action of a two-parameter group of scaling symmetries, namely 
\beq\label{scale}  
\tau_n \rightarrow \rho \, \sigma^n \, \tau_n, \qquad \rho, \sigma \neq 0,    
\eeq 
while in the even case there is the additional symmetry 
\beq\label{extra} 
\tau_n\rightarrow \xi^{(-1)^n}\, \tau_n, \qquad \xi \neq 0. 
\eeq 
In the theory of tau-functions such scalings are known as gauge transformations. 
By Lemma 2.7 in \cite{fordyhonecluster} (see also section 6 therein), symplectic coordinates 
are obtained by taking a complete set of invariants for these scaling symmetries. 
With a suitable choice of coordinates, denoted below by $y_n$, the associated symplectic maps 
can be written in the form of recurrences, like so: 
\begin{itemize}
\item {\bf Odd $d$:} The quantities $y_n =\tau_{n+2}\tau_n/\tau_{n+1}^2$ are 
invariant under (\ref{scale}), and satisfy the difference equation 
\beq\label{dodd} 
y_{n+d-1} \, (y_{n+d-2}\, \cdots\,  y_{n+1})^2 \, y_n =\alpha_n \, y_{n+d-2}\, \cdots \, y_{n+1} +1. 
\eeq 
\item  {\bf Even $d$:} The quantities $y_n =\tau_{n+3}\tau_n/(\tau_{n+2}\tau_{n+1})$ are 
invariant under both (\ref{scale}) and (\ref{extra}), and satisfy the difference equation 
\beq\label{deven} 
y_{n+d-2}\, \cdots \, y_n =\alpha_n \, y_{n+d-3}y_{n+d-5}\, \cdots \, y_{n+1} +1. 
\eeq 
\end{itemize}

\begin{example} \label{pqrt} 
{\em Both  of the cases $d=3$ and $d=4$ lead to iteration of maps in the plane with coefficients that vary periodically, 
namely 
$$ 
\qquad \,\,\, 
d=3: \qquad 
y_{n+2}\, y_{n+1}^2\, y_n = \alpha_n\, y_{n+1} +1, \qquad \mathrm{with} \quad \alpha_{n+2}=\alpha_n ,
$$ 
$$ 
\mathrm{and} \quad  
d=4: \qquad y_{n+2}\, y_{n+1}\, y_n = \alpha_n\, y_{n+1} +1, \qquad \mathrm{with} \quad \alpha_{n+3}=\alpha_n . 
$$ 
}
\end{example}
\begin{remar}\label{nonqrt} 
Due to the presence of the periodic coefficients, the latter maps are not of standard QRT type \cite{qrt1},  
but they reduce to symmetric QRT maps when $\alpha_n=$constant.  
Maps in the plane of this more general type have recently been studied systematically by Roberts \cite{nonqrt}. 
\end{remar} 

In general, the solutions of (\ref{dodd}) or (\ref{deven}) correspond to the iterates of a symplectic map, in 
dimension $d-1$ or $d-2$, respectively. To be more precise, rather than just iterating a single map, 
each iteration depends on the coefficient $\alpha_n$ which varies with a fixed period 
(as in Proposition \ref{somos}), but the same symplectic structure is preserved at each step.   
The appropriate symplectic form in the coordinates $y_n$ can be obtained 
from (\ref{pre}), and by a direct calculation the associated nondegenerate  bracket  
is found; this is  presented as follows. 

\begin{lemm}\label{ybrs} 
For $d$ odd, each iteration of (\ref{dodd}) preserves the nondegenerate Poisson bracket specified by 
\beq\label{yodd} 
\{y_i,y_j\} = (-1)^{j-i+1}\, y_i\, y_j, \qquad 0\leq i<j\leq d-2.  
\eeq 
For $d$ even, the map defined by (\ref{deven}) preserves the nondegenerate bracket in dimension $d-2$ given by 
\beq\label{yeven} 
\{y_i,y_{i\pm 1} \} = \pm \, y_i\, y_{i\pm 1}, 
\eeq 
with all other brackets $\{y_i,y_j\}$  for  $0\leq i<j\leq d-3$  being zero. 
\end{lemm}

The bracket for the variables $y_i$ is the key to deriving a second Poisson structure 
for the KdV maps when 
$d$ is odd. 

\begin{theore}\label{kdvbr2} 
In the case that $d$ is odd, the map (\ref{E:kdvmap}) preserves a second Poisson bracket, which is 
specified  in terms of the coordinates (\ref{gs}) by 
\beq\label{gbr2} 
\{g_i,g_j\}_2 = (-1)^{j-i+1}\, g_i\, g_j, \qquad 0\leq i<j\leq d-1. 
\eeq  
\end{theore} 
\begin{proof}Substituting for $v_n$ from the formula (\ref{tauf}) and making use of the bilinear 
equation (\ref{bil}) produces the identities  
$$ 
v_{j-1}-\frac{1}{v_j}
=\frac{\tau_j(\tau_{j+1}\tau_{j+d-2}-\tau_{j-1}\tau_{j+d})}{\tau_{j+d-1}\tau_{j-1}\tau_{j+1}} = 
-\frac{\alpha_{j-1}\tau_j^2}{\tau_{j-1}\tau_{j+1}}, 
$$ 
so that, by the definition of the symplectic coordinates $y_j$,   
$g_j=- \alpha_{j-1}/y_{j-1}$ for $j = 1, \ldots , d-1$. 
A similar calculation in terms of tau-functions also yields 
$g_0 = - y_0\, y_1 \cdots y_{d-2}$. 
Noting that the coefficients $\alpha_i$ 
play the role of constants 
with respect to the bracket (\ref{yodd}), 
this immediately implies that  the  induced  brackets between the $g_i$ are    
$$\{g_i,g_j\}=\alpha_{i-1}\alpha_{j-1}y_{i-1}^{-2}y_{j-1}^{-2}\{y_{i-1},y_{j-1}\}=(-1)^{j-i+1}g_ig_j$$ 
for $1\leq i< j\leq d-1$, which agrees with (\ref{gbr2}).  By making use of the preceding 
formula for $g_0$ in terms of $y_i$,  the brackets $\{g_0,g_j\}$ follow in the same way.  
\end{proof}
\begin{remar}\label{2cas} 
The bracket (\ref{gbr2}) is the same as the quadratic bracket for the dressing chain 
(for the case where all parameters $\beta_i$ are zero in \cite{shabves}). It has the Casimir 
\beq\label{cas2} 
{\cal C}^* = g_0 \, g_1 \cdots g_{d-1}, 
\eeq 
which is also given in terms of the coefficients of the bilinear equation (\ref{bil}) by 
$  {\cal C}^* =-\prod_{j=0}^{d-2}\alpha_j$. 
\end{remar}

The case where $d$ is even is slightly more complicated, because we do not have a direct way to derive a second 
Poisson bracket in terms of the coordinates $g_i$ (or equivalently $v_i$). The reason for this difficulty 
is that, from the tau-function expressions, although the quantities $g_i$ remain the same under the action 
of the two-parameter symmetry group (\ref{scale}), they are not invariant under the 
additional scaling (\ref{extra}), but rather they transform differently according to the parity of 
the index $i$:
\beq\label{gscale} 
g_i\longrightarrow \xi ^{\pm 4} \, g_i \qquad \mathrm{for} \quad \mathrm{even/odd} \quad i. 
\eeq 
In order to obtain a fully invariant set of variables, we introduce 
a projection $\pi$ from dimension $d$ to dimension $d-1$:  
$$ 
\pi: \quad f_i = g_ig_{i+1}, \qquad i = 0, \ldots , d-2.       
$$ 
By regarding the new coordinates $f_i$ as functions of the $g_j$, we get an induced map 
$\varphi '$ in dimension $d-1$, which is compatible with $\varphi$ 
in the sense that $\varphi  \cdot \pi = \pi \cdot \varphi '$; this has the form 
\beq\label{phip} 
\bear{rl} 
\varphi ': & \quad (f_0,f_1, \ldots , f_{d-2})  \\ 
& \mapsto \left(  f_1(1+f_0^{-1}), f_2, \ldots , f_{d-2}, \frac{f_0 f_2\cdots f_{d-2}}{(1+f_0^{-1})f_1f_3\cdots f_{d-3}}\right) . 
\eear 
\eeq 
In terms of tau-functions, the quantities $f_i$ are 
invariant under the action of the full three-parameter group of gauge transformations, which means 
that they can be expressed as functions of the invariant symplectic coordinates $y_n$, and hence we can 
derive a Poisson bracket for them. 

\begin{theore}\label{fbr2} 
In the case that $d$ is even, the map (\ref{phip}) preserves a Poisson bracket, which is 
specified  in terms of the coordinates $f_i$  by 
\beq\label{brf2} 
\{f_i,f_{i\pm 1}\}_2 = \pm \, f_i\, f_{i\pm 1},  
\eeq  
with all other brackets $\{f_i,f_j\}_2$  for  $0\leq i<j\leq d-2$  being zero. 
\end{theore} 
\begin{proof}
Following the proof of Theorem \ref{kdvbr2}, we have 
$g_ig_{i+1} = \alpha_{i-1}\alpha_i\tau_i\tau_{i+1}/(\tau_{i-1}\tau_{i+2})$, 
whence, in terms of the symplectic coordinates for  (\ref{deven}), we have 
$ f_i =  \alpha_{i-1}\alpha_i/y_{i-1}$ for $i =1, \ldots , d-2$. 
A similar calculation yields a slightly different formula for index $i=0$: 
$f_0=\alpha_0\, y_1y_3\cdots y_{d-3}$, 
and then  the Poisson brackets   (\ref{yeven}) 
between the $y_j$  directly imply that the brackets (\ref{brf2}) hold between the $f_i$. 
\end{proof} 
\begin{remar} \label{2fcas} 
The bracket (\ref{brf2}) has the Casimir 
\beq\label{casf2} 
{\cal C}^* = f_0 \, f_2 \cdots f_{d-2}, 
\eeq 
which is also given in terms of the coefficients of the bilinear equation (\ref{bil}) by 
$  {\cal C}^* =\prod_{j=0}^{d-2}\alpha_j$. 
\end{remar}

It may be unclear why there is the suffix $2$ on the bracket in  (\ref{brf2}), since we have not yet 
provided another Poisson bracket for the map (\ref{phip}). However, as will be explained in the next section, 
when $d$ is even the quantities $f_i$ form a Poisson subalgebra for the bracket $\{\,, \,\}_1$ of Theorem \ref{vbracket1}. 
As will also be explained, the brackets $\{\, , \,\}_1$ and  $\{\, , \,\}_2$ are compatible (in the sense that any linear 
combination of them is also a Poisson bracket), which means that a standard bi-Hamiltonian argument  can be used 
to show Liouville integrability of the maps for either odd or even $d$.

\section{First integrals and integrability of the KdV maps} 


The purpose of this section is to prove the following result. 

\begin{theore}
\label{kdvint} 
For each $d\geq 3$, 
the map (\ref{E:kdvmap}) is completely integrable in the Liouville sense.  
\end{theore} 

Recall that, for a Poisson map $\varphi$ in dimension $d$, Liouville integrability means that there should be 
$k$ Casimirs invariant under the map, so that the $(d-k)$-dimensional symplectic leaves 
of the Poisson bracket are 
preserved by $\varphi$, plus 
an additional $\frac{1}{2}(d-k)$ independent first integrals that are in involution 
with respect to the bracket. For the particular map in question, there is always the Poisson bracket $\{ \, , \,\}_1$, and from  
Lemma \ref{gcoords} this has  either one or two Casimirs, with symplectic leaves 
of dimension $d-1$ or $d-2$, for odd/even $d$ respectively; hence an additional 
$\lfloor \frac{1}{2}(d-1)\rfloor$ first integrals are required in this case. 
 

For the lowest values $d=3,4$, it is straightforward to verify complete integrability, since 
in those cases only one extra first integral is required, apart from the Casimirs of $\{\, ,\,\}_1$.  
When $d=3$, as in Example \ref{d3br1},  the Casimir $\cal C$  
is preserved by $\varphi$; the quantity ${\cal C}^*$ in (\ref{cas2}), which is 
the Casimir of the second bracket, is also a first integral. 
Similarly, for $d=4$, the two Casimirs in Example 
\ref{d4br1} provide the two first integrals  (\ref{K}), which are themselves Casimirs;  
the quantity ${\cal C}^* = f_0f_2=g_0g_1g_2g_3$, as in (\ref{casf2}), provides 
the extra first integral.   

\begin{example}\label{d3curve} 
{\em 
For $d=3$, each of the first integrals ${\cal C}$,  ${\cal C}^*$ of the map  (\ref{E:kdvmap}) 
define surfaces in 
three dimensions,  given by 
$$ 
\bear{rcl} 
XY^2Z- {\cal C}\, XYZ - XY-YZ-ZX& =& 0, \\ 
XY^2Z + {\cal C}^*\,XYZ-XY-YZ+1 & =& 0, 
\eear 
$$ 
respectively, in terms of 
coordinates $(X,Y,Z)\equiv (v_n,v_{n+1},v_{n+2})$. 
These two surfaces intersect in a curve of genus one, 
found 
explicitly by eliminating  the variable $Z$  above; this yields the biquadratic 
\beq\label{biquad} 
 ({\cal C}+ {\cal C}^*)\, X^2Y^2 +(X+Y)(XY-1)- {\cal C}\, XY =0. 
\eeq 
The embedding of such a curve in three dimensions can be seen from the orbit plotted 
in Figure 1. 
} 
\end{example} 

%
%
%
%

For higher values of $d$, it is not so obvious how 
to proceed, but the correct number 
of additional first integrals can be obtained by constructing a Lax pair for the map, as we now describe.  
  
\subsection{Lax pairs and monodromy} 

The discrete KdV equation (\ref{E:kdv}) is known to be integrable in the sense that it arises as the 
compatibility condition for a pair of linear equations. 
In   \cite{Nijhoff2006lecture}, a scalar Lax pair is given as follows: 
$$ 
\widehat{\widetilde{\phi}} =  v \widetilde{\phi}+\lambda \phi, 
\qquad 
\widehat{\phi}=\widetilde{\phi}+\frac{1}{v}\phi . 
$$ 
Upon introducing the vector $\Phi:= 
\left(\begin{array}{c} \widetilde{\phi} \\{\phi}\end{array}\right)$, 
the latter 
pair of scalar equations  
leads to a Lax pair  in matrix form, 
namely 
\begin{equation}
\label{E:Laxpairs}
\widetilde{\Phi}={\bf L}\Phi, \qquad \widehat{\Phi}={\bf M}\Phi, \qquad \mathrm{with} \quad {\bf L}=\begin{pmatrix}
v-\frac{1}{\widetilde{v}}&\lambda\\
1& 0
\end{pmatrix}, \quad  
{\bf M}=\begin{pmatrix}
v& \lambda\\
1&\frac{1}{v}
\end{pmatrix}.
\end{equation}
Equation~(\ref{E:kdv}) is equivalent to the compatibility condition   for 
the linear system (\ref{E:Laxpairs}), that is 
\beq\label{laxe} 
\widehat{{\bf L}}\, {\bf M}= \widetilde{{\bf M}}\, {\bf L}.
\eeq 
Note that the matrix ${\bf L}$ is associated with shifts in the $\ell$  (horizontal) direction, 
and ${\bf M}$ with shifts in the $m$  (vertical) direction. 

First integrals of each KdV map~\eqref{E:kdvmap} can be found by using the  staircase method 
\cite{CapelNijhoff1991, QCPN1991, Kamp2009Initialvalue, KampStaircase}. For the $(d-1,-1)$-reduction, 
a staircase on the $\mathbb{Z}^2$ lattice is built from paths consisting of 
$d-1$ horizontal steps and one vertical step. 
Taking an ordered product of Lax matrices along the staircase yields the monodromy matrix
\begin{equation}
\label{E:mono}
\mathcal{{ L}}_n={\bf M}_n^{-1}\, {\bf L}_{n+d-2}\cdots {\bf L}_n,   
\end{equation}
corresponding to $d-1$ steps to the right ($\ell \to \ell +1$) and one step down ($m\to m-1$), 
where ${\bf L}\to{\bf L}_n$ and 
${\bf M}\to{\bf M}_n$ under the reduction. 
As a consequence of (\ref{laxe}), the identity ${\bf L}_{n+d-1}{\bf M}_n = {\bf M}_{n+1}{\bf L}_n$ holds, 
which implies that $\mathcal{ L}_n$ satisfies the discrete Lax equation 
\beq \label{dlax} 
\mathcal{L}_{n+1} \, {\bf L}_n= {\bf L}_n \, \mathcal{L}_{n}. 
\eeq 
The latter holds if and only if  $v_n$ satisfies the difference equation 
(\ref{E:KdV_Reduction}). Since (\ref{dlax}) means that the spectrum of  $\mathcal{ L}_n$  is 
invariant under the shift $n\to n+1$, first integrals   for the KdV map can 
be constructed from the trace of the monodromy matrix (or powers thereof), which can be 
expanded in the spectral parameter $\lambda$.

For convenience, we conjugate $\mathcal{ L}_n$ in (\ref{E:mono}) by ${\bf M}_n$ and multiply by an overall 
factor of $\lambda -1$, 
which (upon setting $n=0$) gives a modified monodromy matrix 
\beq\label{momo} 
\mathcal{L}^\dagger = {\bf L}^\dagger_{d-1}\cdots {\bf L}^\dagger_1 \, {\bf L}^\dagger_0 ,  
\eeq 
where 
$$ 
{\bf L}^\dagger_{i} ={\bf L}_{i-1} = \left( \begin{array}{cc} g_{i} & \lambda \\ 1 & 0 \end{array}\right), \, 
 i=1, \ldots , d-1, 
\quad  \mathrm{but} \,\,   
{\bf L}^\dagger_0 =  \left( \begin{array}{cc} g_{0} & \lambda \\ 1 & \frac{\nu}{g_0} \end{array}\right) 
\, \mathrm{with} 
\, \, \nu =1.
$$ 
Thus we see the origin of the coordinates $g_i$ introduced in Lemma \ref{gcoords}: they are the $(1,1)$ entries 
in the Lax matrices that make up the monodromy matrix. 
We have already mentioned a connection with  the 
dressing chain at the level of the Poisson brackets, but it can be seen more directly here: up to taking an inverse 
and inserting a spectral parameter, when $\nu =0$ 
the expression (\ref{momo}) reduces to the monodromy  matrix for the dressing chain found in    
\cite{fordyhonecluster}. The parameter $\nu$ can be regarded as introducing inhomogeneity 
into the chain at $i=0$ (and for $\nu\neq 0$ it is always possible to rescale so that $\nu=1$).   

With the introduction of another spectral parameter $\mu$, we consider the characteristic polynomial 
of   $\mathcal{L}^\dagger$, which defines the spectral curve 
\beq\label{spec} 
\chi (\lambda , \mu):=\det (\mathcal{L}^\dagger -\mu \, 1) = \mu^2 - {\cal P}(\lambda )\, \mu 
+ (\nu -\lambda)(-\lambda)^{d-1}=0, 
\eeq 
with ${\cal P}(\lambda )= \mathrm{tr}\, \mathcal{L}^\dagger$. This curve in the $(\lambda ,\mu )$ plane 
is invariant under the KdV map  (\ref{E:kdvmap}), 
and the trace of  $\mathcal{L}^\dagger$ provides the polynomial  ${\cal P}$ whose 
non-trivial coefficients are first integrals of the map. 

\begin{example} \label{d3spec}
{\em 
When $d=3$ the spectral curve is of genus one, given by 
$ 
\mu^2 -({\cal C}^*+{\cal C}\lambda)\, \mu +(\nu -\lambda )\lambda^2=0. 
$ 
For $\nu =1$ this is isomorphic to the biquadratic curve 
(\ref{biquad}) corresponding to the intersection of the level sets of 
${\cal C}$ and ${\cal C}^*$. 
} 
\end{example}
\begin{example} \label{d4spec}
{\em 
When $d=4$ the spectral curve also has   genus one, being given by 
$
\mu^2 -(I_0+ I_1\, \lambda +2\lambda^2)\, \mu -(\nu -\lambda )\lambda^3=0$,  
with coefficients expressed in terms of $g_i$ by 
$$ 
 I_0 = {\cal C}^* = g_0g_1g_2g_3, 
\qquad 
 I_1 = {\cal K}'-\nu = g_0g_1+g_1g_2+g_2g_3+g_3g_0+\nu\, \frac{g_2}{g_0}. 
$$
} 
\end{example} 

In the next subsection, we will show how, for each $d$, expanding the trace of the monodromy matrix in powers of 
$\lambda$ gives the polynomial ${\cal P}$ of degree $\lfloor d/2\rfloor$. This implies that the hyperelliptic curve 
$\chi (\lambda ,\mu )=0$   as in (\ref{spec}) has genus $\lfloor \frac{1}{2}(d-1)\rfloor$. Thus we expect 
that the (real, compact) Liouville tori for the KdV map  should be identified with (a real component of) the 
Jacobian variety of this curve, since their dimensions coincide, and the map should correspond to a 
translation on the Jacobian.

\subsection{Direct proof from the Vieta expansion} 

In order to calculate the trace of the monodromy matrix explicitly, we split the Lax matrices as 
\begin{equation}
\label{E:spliting}
{\bf L}^\dagger_i=\lambda \begin{pmatrix}
0&1\\
0&0
\end{pmatrix}
+
\begin{pmatrix}
g_i&0\\
1&0
\end{pmatrix}
=\lambda \, {\bf X}+ {\bf Y}_i, \qquad i=1, \ldots d-1, 
\end{equation}
and similarly for ${\bf L}^\dagger_0$. 
This splitting fits into the framework of \cite{Tranclosedform},  with 
the coefficient of $\lambda$ being the same nilpotent matrix ${\bf X}$ in each case, and allows 
the application of the so-called Vieta expansion. 
  
Recall that, in the noncommutative setting, the Vieta expansion
is given by the formula 
\beq\label{vieta} 
\prod_{i=a}^{\overset{\curvearrowleft}{b}}(\lambda {\bf X}_i+{\bf Y}_i):= 
(\lambda {\bf X}_b+{\bf Y}_b)\ldots(\lambda {\bf X}_a+{\bf Y}_a)     
=\sum_{r=0}^{b-a+1}\lambda^{r}{\bf Z}^{a,b}_r,  
\eeq 
 where
$$ 
{\bf Z}^{a,b}_r  
= \sum_{a\leq i_1 < i_2 < \cdots < i_r \leq b}{\bf  Y}_b{\bf Y}_{b-1}\cdots
{\bf Y}_{i_r+1}{\bf X}_{i_r}{\bf Y}_{i_r-1} \cdots {\bf Y}_{i_1+1}{\bf X}_{i_1}{\bf Y}_{i_1-1} \cdots {\bf Y}_a  .
$$
In the case at hand, we have ${\bf X}_i ={\bf X}$ for all $i$, and 
by writing the monodromy matrix  in the form  (\ref{vieta}), 
we can 
use Lemma 8 in \cite{Tranclosedform} 
to expand 
the trace 
as 
\beq\label{trace} 
{\cal P}(\lambda)={\mathrm tr} \,  \mathcal{L}^\dagger= \sum_{r=0}^{\lfloor d/2\rfloor}I_r\, \lambda^r, 
\eeq 
where the first  integrals of the KdV map are given by the closed-form expression 
\beq
\label{E:KdVInv}
I_r=\Psi^{1,d-3}_{r-1}+g_0\Psi^{1,d-4}_{r-1}+\left(g_{d-1}+\frac{\nu}{g_0}\right)\Psi^{2,d-3}_{r-1} 
+\Psi^{2,d-4}_{r-2}+g_0g_{d-1} \Psi^{1,d-3}_r,
\eeq 
with 
\begin{equation}
\label{D:psi}
\Psi^{a,b}_r=\prod_{i=a}^{b+1}g_i\, \left(\sum_{a\leq i_1<i_1+1<i_2\ldots<i_r\leq b} \prod_{j=1}^r\frac{1}{g_{i_j}g_{i_j+1}}\right) . 
\end{equation}
The explicit form of these integrals mean that it is possible to give a direct verification that they are  
in involution, which is almost identical to the proof of Theorem 13 in \cite{TranInvo}. 
\begin{propositio}
\label{T:invo}
The first  integrals (\ref{E:KdVInv}) of the KdV map Poisson commute with respect to  the bracket $\{\, ,\,\}_1$ in~\eqref{E:KdVsymp}.
\end{propositio}
\begin{proof} 
Writing the Poisson  structure in terms of the coordinates $g_i$ as in Lemma \ref{gcoords}, 
we see that (up to an overall factor of  $-1$) the brackets $\{g_i,g_j\}_1$ in (\ref{C6E:KdVPoi}) 
for $1\leq i,j\leq d-2$ are identical to the brackets between the coordinates $c_i$ given by 
equation (46) in  \cite{TranInvo}. The polynomial functions  given by (\ref{D:psi}) 
are the same as in the latter reference, and the particular functions $\Psi^{a,b}_r$ 
that appear in the formula (\ref{E:KdVInv}) only depend on $g_i$ for  $1\leq i\leq d-2$. 
This implies that the brackets between these functions are the same as those given in Lemma 11 and Corollary 12 in 
\cite{TranInvo}, whence it is straightforward to verify that $\{I_r,I_s\}_1=0$ for $0\leq r,s\leq \lfloor d/2\rfloor$. 
\end{proof} 
In order to complete the proof of Theorem \ref{kdvint}, it only remains to check that there are sufficiently 
many independent integrals. 

For odd $d$, the leading ($\nu =0$) part of the first integral $I_r$ is a cyclically symmetric,  homogeneous 
function of degree $d-2r$ in the $g_i$, being identical to the independent integrals of the dressing chain in \cite{shabves} 
(for parameters $\beta_i=0$), with terms of each odd degree appearing for $r=0,\ldots , \frac{d-1}{2}$. 
Hence these functions are also  independent in the case $\nu \neq 0$,  
corresponding to the addition of a rational term with $g_0$ in the denominator, which is 
linear in $\nu$  and  of  degree  $d-2r-2$ in the $g_i$. This means that there are $\deg {\cal P}= (d-1)/2 +1$ 
independent integrals, of which the last one is $I_{(d-1)/2}={\cal C}$, the Casimir of the bracket $\{\, ,\,\}_1$, 
as given in (\ref{oddc}). 

When $d$ is even, the leading  ($\nu =0$) part of $I_r$ has the same structure, with terms of 
each even degree $d-2r$ appearing for $r=0,\ldots ,\frac{d}{2}-1$, but for all $\nu$ the last one is trivial: 
$I_{d/2}=2$. Thus the coefficients of $\cal P$ provide only $\frac{d}{2}$ independent first integrals.  
The last non-trivial coefficient is a Casimir of $\{\, ,\,\}_1$, given by $I_{d/2-1}={\cal C}_1 {\cal C}_2-\nu$, 
in terms of the two Casimirs in (\ref{evenc}), and the Casimir ${\cal C}_1+{\cal C}_2$ 
provides one extra first integral, as required. 

In the next subsection, we 
outline  another proof of  Theorem \ref{kdvint}, based on the 
bi-Hamiltonian structure obtained from the two Poisson brackets 
$\{\, ,\,\}_1$ and $\{\, ,\,\}_2$.

\subsection{Proof via the bi-Hamiltonian structure}  

As an alternative to the direct calculation of brackets between the first 
integrals  (\ref{E:KdVInv}),  the two Poisson brackets can be used to show complete 
integrability, with a suitable Lenard-Magri chain; 
this is a standard method for bi-Hamiltonian systems \cite{magri}. 
This approach applies immediately to the coordinates $g_i$ when $d$ is odd and the brackets are given by  
(\ref{C6E:KdVPoi}) and  (\ref{gbr2}); but  minor modifications are required to apply it to the coordinates 
$f_i$ when $d$ is even. 

For $d$ odd, the first observation to make is that,  
just as in the case of the dressing  chain ($\nu =0$),  
the brackets  $\{\, ,\,\}_1$ and $\{\, ,\,\}_2$ are compatible with each other, 
in the sense that their sum, and hence any linear combination of them, is also a Poisson bracket. 
Thus these two compatible brackets form a bi-Hamiltonian structure for the KdV map $\varphi$. 
A pencil of Poisson brackets is defined by the bivector field 
${\tt P}_2-\lambda \, {\tt P}_1$, where ${\tt P}_j$ denotes the bivector corresponding to 
$\{\, ,\,\}_j$ for $j=1,2$. (In fact, since it defines a  bracket for all $\lambda$ and $\nu$, 
this gives {\it three} compatible Poisson structures.) 
Then, by a minor adaptation of Theorem 4.8 in   \cite{fordyhonecluster} (which, up to rescaling 
by a factor of 2, is the special case $\lambda =-1$ and $\nu =0$) we see that 
the trace of the monodromy matrix 
is a 
Casimir of the Poisson pencil, or in other words 
\beq \label{caspen} 
({\tt P}_2-\lambda \, {\tt P}_1)\, \lrcorner \, \dd {\cal P}(\lambda )=0. 
\eeq
Expanding this identity in powers of $\lambda$ yields a finite 
Lenard-Magri chain, starting with $I_0={\cal C}^*$,  
the Casimir of the bracket $\{\, ,\,\}_2$ (as in Remark \ref{2cas}), and ending with 
$I_{(d-1)/2}={\cal C}$, the Casimir of the bracket $\{\, ,\,\}_1$ (as in Lemma \ref{gcoords}): 
$
 {\tt P}_2\, \lrcorner \, \dd I_0 =0, 
{\tt P}_2\, \lrcorner \, \dd I_1 = {\tt P}_1\, \lrcorner \, \dd I_0 , 
\ldots 
 {\tt P}_2\, \lrcorner \, \dd I_r = {\tt P}_1\, \lrcorner \,  \dd I_{r-1} , 
\ldots   
 {\tt P}_1\, \lrcorner \, \dd I_{(d-1)/2}=0
$. 
It then follows, by a standard inductive argument, that 
$\{I_r,I_s\}_1=0=\{I_r,I_s\}_2$ for $0\leq r,s\leq  (d-1)/2$. 
Hence we see that Proposition \ref{T:invo} is a consequence 
of the bi-Hamiltonian structure in this case. 

In the case where $d$ is even, a slightly more indirect argument is necessary, making the projection 
$\pi$ and working with the map $\varphi '$ in dimension $d-1$, as given by (\ref{phip}). 
The main ingredient required is 
the expression for the relations between the coordinates $f_i$ for  $i=0,\ldots ,d-2$ 
with respect to the first bracket. The case $d=4$ is special, so we do this example first before 
summarizing the general case. 
\begin{example}\label{fdeven} 
{\em For $d=4$, we use Lemma \ref{gcoords} to calculate 
$$
\bear{rcl}  
\{ f_0,f_1\}_1 = \{g_0g_1,g_1g_2\}_1 & = &\{g_0,g_1\}_1g_1g_2 +g_1(\{g_0,g_2\}_1g_1 +g_0\{g_1,g_2\}_1) 
\\ 
&=& -g_1g_2-g_0g_1=-f_{0}-f_1, 
\eear 
$$ 
and similar calculations show that 
$$ 
\{f_0,f_2\}_1=f_1-\frac{f_0f_2}{f_1}+\nu\, \frac{f_1}{f_0}, \qquad 
\{f_1,f_2\}_1=-f_1-f_2+\nu\, \frac{f_1^2}{f_0^2}, 
$$ 
which implies that $f_0,f_1,f_2$ generate a three-dimensional Poisson subalgebra 
for the bracket $\{\, ,\,\}_1$. It follows that (\ref{phip}) is a Poisson map 
(in three dimensions) with respect to the restriction of this bracket to the subalgebra. 
Moreover, the restricted bracket $\{\, ,\,\}_1$ for the $f_i$ is compatible with 
the bracket $\{\, ,\,\}_2$ given by (\ref{brf2}) with $d=4$.  
The coefficients of the corresponding spectral curve, as in Example \ref{d4spec}, 
can be written as functions of the $f_i$, i.e. 
$$ 
I_0=    f_0f_2, \qquad I_1 = f_0+f_1+f_2+\frac{f_0f_2}{f_1}+\nu\,\frac{f_1}{f_0}, 
$$ 
and these generate the short Lenard-Magri chain 
$ \{ \, \cdot\, , I_0\}_2 =0$, 
$\{ \,\cdot \, , I_1\}_2 =\{ \, \cdot\,  , I_0\}_1$, 
$\{\,  \cdot \, , I_1\}_1 =0$. 
} 
\end{example}  
\begin{theore}
\label{fcoords} 
For all even $d\geq 6$, the first Poisson structure for the KdV map (\ref{E:kdvmap}) 
reduces to a bracket for the coordinates $f_i=g_i\, g_{i+1}$, 
as specified by the 
following relations for $\nu =1$ and $0\leq i<j\leq d-2$:
\begin{equation}
\label{f1}
\{f_{i},f_{j}\}_1=\left\{
\begin{array}{ll}
-f_i-f_{i+1}& \mbox{if}\  \, j-i=1,\\  
\\
-\frac{f_if_{i+2}}{f_{i+1}} &\mbox{if}\ \, j-i=2,\\  
\\
\frac{f_1f_3\cdots f_{d-3}}{f_2f_4\cdots f_{d-4}}\left(1+\frac{\nu}{f_0} \right) & \mbox{if}\ j-i=d-2,\\ 
\\
\nu \, \frac{f_1^2f_3\cdots f_{d-3}}{f_{0}^2f_2\cdots f_{d-4}}& \mbox{if}\ i=1, j=d-2,\\\ 
\\
0&\mbox{otherwise}.
\end{array}
\right.
\end{equation}
This Poisson bracket is  preserved by the map (\ref{phip}).  
\end{theore}

The rest of the argument proceeds as for $d$ odd: the brackets 
$\{\, ,\,\}_1$ and $\{\, ,\,\}_2$  (as given in (\ref{f1}) and (\ref{brf2}) respectively) 
are compatible with each other, so they provide a bi-Hamiltonian structure for 
$\varphi '$ in dimension $d-1$. Letting   ${\tt P}_j '$ for $j=1,2$ denote the corresponding 
Poisson  bivector fields, the analogue of (\ref{caspen}) can then be verified: 
$ ({\tt P}_2 '-\lambda \, {\tt P}_1 ')\, \lrcorner \, \dd {\cal P}(\lambda )=0$. 
It should be noted that the latter identity is well-defined in the $(d-1)$-dimensional 
space with coordinates $f_i$: the trace of the monodromy matrix is invariant 
under the scaling symmetry  (\ref{gscale}), hence all of the first integrals $I_r$ 
can be written as functions of the variables $f_i$. 
In this case, the Lenard-Magri chain begins with  $I_0={\cal C}^*$,  
the Casimir of $\{\, ,\,\}_2$ (as in Remark \ref{2fcas}), and ends with 
a Casimir of  the bracket $\{\, ,\,\}_1$, namely    
$I_{d/2-1}={\cal C}_1 {\cal C}_2-\nu$ 
(which is well-defined in terms of $f_i$). 
Thus it follows that the integrals $I_r$ are in involution 
with respect to both brackets for the  $f_i$, and this result extends to the bracket $\{\, ,\,\}_1$ 
when it is lifted to $d$ dimensions (in terms of $g_i$, or equivalently $v_i$). 

\section{First integrals and integrability of the double pKdV maps}
In this section, we go back to the  $(d-1,-1)$-reduction of the 
double pKdV equation (\ref{double}). 
This reduction yields the difference equation (\ref{E:2copies}), 
and in section 2 we showed how writing this as 
the discrete Euler-Lagrange equation (\ref{E:Eu-LagrEq.})  
led to a symplectic structure for the corresponding $2d$-dimensional map (\ref{E:pKdVmap}).
It is easy to see that the first integrals~\eqref{E:KdVInv} of the KdV map 
also provide first 
integrals for equation~\eqref{E:2copies}
by writing $v_{i}=u_{i-d}-u_{i}$ for $i=0,\ldots ,d-1$. However,  the total number of independent 
integrals for (\ref{E:kdvmap})  is only 
$\lfloor d/2\rfloor +1 $, which is not enough for complete
integrability of the $2d$-dimensional symplectic map (\ref{E:pKdVmap}). 
Here we will show how to construct sufficiently many 
additional integrals, leading to a proof of the following result. 

\begin{theore}\label{pkdvli} 
For all $d\geq 3$, the  $2d$-dimensional map  given by (\ref{E:pKdVmap}) with (\ref{E:2copies}) is 
completely integrable in the Liouville sense.  
\end{theore}  

The proof of the above, given  in subsection 5.1 below,  
relies on the 
observation 
that equation~\eqref{E:2copies} can be rewritten as 
$$ 
u_{n+1}-u_{n+d-1}-\frac{1}{u_n-u_{n+d}}=u_{n-d+1}-u_{n-1}-\frac{1}{u_{n-d}-u_n}, 
$$ 
from which we see that 
\begin{equation}
\label{E:gQRT}
u_{n-d+1}-u_{n-1}-\frac{1}{u_{n-d}-u_{n}}=J_n, \qquad \mathrm{with}\quad J_{n+d}=J_n.  
\end{equation} 
Thus the function 
$J_n$  is a $d$-integral (in the sense of \cite{haggar}) for the 
double pKdV map 
(\ref{E:pKdVmap}): for any $n$ it can be viewed as a function on the $2d$-dimensional phase space 
with coordinates $(u_{-d},\ldots ,u_{d-1})$, and under shifting $n$ it is 
periodic  with period $d$.
This implies that any cyclically symmetric function of 
$J_0,J_{1},\ldots, J_{d-1}$ is a first integral 
for equation~\eqref{E:2copies}. As we shall see, 
for the corresponding map (\ref{E:pKdVmap})
this has the further consequence 
that it is {\it superintegrable}, in the sense 
that it has more than the number of independent integrals 
required for Liouville's theorem.

\begin{remar}\label{s12} 
For the $(d-1,-1)$-reduction, 
the 
observation that (\ref{E:gQRT}) holds 
can be seen as 
a direct consequence 
of the fact that (\ref{double}) can be rewritten as (\ref{doublej}). 
An analogous observation applies to the $(s_1,-s_2)$-reduction 
of (\ref{double}), where $n=s_2\ell +s_1 m$, and in  (\ref{doublej}) one has 
$\mathrm{J}\to J_n$ with $J_{n+s_1+s_2}=J_n$. 
\end{remar} 

\subsection{Construction of integrals in involution} 
 
In order to construct additional integrals, we need to calculate the Poisson brackets between the $J_i$, as well as their brackets with 
the $v_j$; together these generate a Poisson subalgebra for the bracket in Theorem \ref{ubracket}, as 
described by the next two lemmata. 

\begin{lemm}
\label{L:PoiJandV} 
With respect to the nondegenerate Poisson bracket specified by  (\ref{E:Poi1}),  
the quantities $J_i$ defined in (\ref{E:gQRT}) and 
$v_{j}=u_{j-d}-u_j$ Poisson commute, i.e. 
\begin{equation}
\label{E:PoiJandV}
\{J_{i}, v_{j}\}=0, \qquad 0\leq i,j\leq d-1. 
\end{equation}
\end{lemm}
\begin{proof} 
For $0< j < d-1$ we have
$$ 
\begin{array}{rcl} 
\{J_{0}, v_{j}\}& = & \{u_{-d+1}-u_{-1}-v_0^{-1},v_j\} = \{ u_{-d+1}-u_{-1},u_{j-d}-u_j\} +v_0^{-2}\{v_0,v_j\} 
\\ 
& = & -\{u_{-d+1},u_j\} + v_0^{-2}\, (-1)^j\prod_{r=0}^j v_r^2 =0, 
\end{array}  
$$ 
where we have used (\ref{E:KdVsymp}) as well as Theorem \ref{ubracket}. 
Similar calculations show that $\{J_0,v_0\}=0=\{J_0,v_{d-1}\}$, 
and, since  the bracket is preserved by the double pKdV map  (\ref{E:pKdVmap}), 
the vanishing of all the other brackets $\{J_{i}, v_{j}\}$ follows by shifting indices. 
\end{proof}
\begin{lemm}
\label{L:PoiJ}
The Poisson brackets between the $d$-integrals $J_n$ are specified by 
\begin{equation}
\label{E:PoiJ}
\{J_{i},J_{j}\}=\left\{
\begin{array}{ll}
1,&\mbox{if}\ j-i=1,\\
-1,&\mbox{if}\ j-i=d-1,\\
0,& \mbox{otherwise},
\end{array}
\right.
\end{equation}
for $0\leq i<j\leq d-1$. 
\end{lemm}
\begin{proof}
Once again, from the behaviour  under shifting indices, 
it is enough to verify the brackets for $i=0$ and $1\leq j\leq d-1$.
Expanding the left hand side of~\eqref{E:PoiJ}, 
and using  Lemma~\ref{L:PoiJandV}, we obtain
$$ 
\begin{array}{rcl} 
\{J_0,J_j\} & = & \{J_0, u_{j-d+1}-u_{j-1}-v_j^{-1}\} \\ 
&
= & \{u_{-d+1}-u_{-1}-(u_{-d}-u_0)^{-1}, u_{j-d+1}-u_{j-1}\} \\  
& = & 
-\{u_{-d+1},u_{j-1}\} +(u_{-d}-u_0)^{-2}\, \Big( \{u_{-d},u_{j-d+1}\} -\{u_{-d},u_{j-1}\}\Big). 
\end{array} 
$$
Upon substituting with the non-zero 
brackets for the coordinates $u_i$,  as in Theorem  \ref{ubracket}, 
the above expression vanishes for $2\leq j\leq d-2$, and is equal to $\pm 1$ for $j=1,d-1$ respectively.  
\end{proof}
The brackets (\ref{E:PoiJ}) mean that a suitable set of quadratic functions 
of the $J_i$ 
give additional integrals for  the equation (\ref{E:2copies}). 
\begin{propositio}\label{tints} 
The functions 
\beq\label{ts} 
T_s=\sum_{i=0}^{d-1} J_{i}\ J_{i+s} 
\eeq 
provide $\lfloor d/2\rfloor + 1$ independent first integrals for the 
double pKdV map  
(\ref{E:pKdVmap}).  
These integrals are in involution with each other, and 
Poisson commute with the first integrals of the 
KdV map (\ref{E:kdvmap}). 
\end{propositio} 
\begin{proof}  
With indices read mod $d$, 
the quantities $T_s$ are cyclically symmetric functions of the $J_i$; in other words they are 
invariant under the cyclic permutation $(J_0,\ldots,J_{d-1})\mapsto (J_{1},\ldots,J_{d-1},J_0)$, 
which means that they are first integrals for (\ref{E:pKdVmap}).  
From the periodicity of the $d$-integrals $J_i$ 
it is clear that $T_{s+d}=T_s$, and also 
$T_{d-s}=T_s$. Taking $s=0,\ldots ,\lfloor d/2\rfloor$ 
yields $\lfloor d/2\rfloor + 1$ independent functions of $J_0,\ldots , J_{d-1}$, and since the $J_i$ are themselves 
independent functions of $u_{-d},\ldots ,u_{d-1}$, this implies functional independence of this number of the quantities (\ref{ts}).  
The fact that these quantities Poisson commute with each other 
is a consequence of Lemma \ref{L:PoiJ}, and the fact that 
$
\frac{\partial T_r}{\partial J_{i}}=J_{i+r}+J_{i-r} 
$, 
which implies that 
\[
\{T_r,T_s\}=\sum_{i=0}^{d-1}\left(\frac{\partial T_r}{\partial J_{i}}\frac{\partial T_s}{\partial J_{i+1}}
-\frac{\partial T_r}{\partial J_{i+1}}\frac{\partial T_s}{\partial J_{i}}\right) =0,
\]
as required.  Also, by Lemma \ref{L:PoiJandV}, each $T_s$ Poisson commutes 
with any function of the $v_i$, so with the 
first integrals of (\ref{E:kdvmap})  in particular.  
\end{proof}

Before describing the general case, we now demonstrate 
complete integrability for the simplest examples. 
\begin{example}\label{6d} 
{\em 
Starting with $d=3$, 
equation~\eqref{E:2copies} yields the $6$-dimensional  symplectic map
\beq 
\label{E:N=2}
\left(u_{-3},u_{-2},\ldots, u_{2}\right)\mapsto \Big(u_{-2},\ldots, u_{2}, F\left(u_{-3},u_{-2},\ldots, u_{2}\right)\Big),
\eeq
where
$ 
F\left(u_{-3},u_{-2},\ldots, u_{2}\right) 
=u_0-\Big( u_{1}+u_{-1}-u_{2}-u_{-2}+1/(u_{-3}-u_0)\Big)^{-1} 
$. 
Two integrals of  this map
are given in terms of $v_0 = u_{-3}-u_0$, $v_1 = u_{-2}-u_1$, $v_2 = u_{-1}-u_2$  by
$$ 
I_0 \equiv  {\cal C}^*=-v_{1}+\frac{1}{v_0}+\frac{1}{v_{2}}-\frac{1}{v_0v_{1}v_{2}}, \qquad 
I_1 \equiv  {\cal C}=v_1 -\frac{1}{v_0}-\frac{1}{v_{1}}-\frac{1}{v_{2}}  
$$ 
Apart from these, there is the pair of integrals 
$T_0 = J_0^2 +J_1^2 +J_2^2$, 
$T_1 = J_0J_1+J_1J_2+J_2J_0$, 
which are written as symmetric functions of the $3$-integrals %
$J_{i}=u_{i-2}-u_{i-1}- 
(u_{i-3}-u_{i})^{-1}$, 
$i=0,1,2$.
The quantities $I_0,I_1,T_0,T_1$ Poisson commute with each other, but they cannot 
all be independent. Indeed, the quantities $v_0,v_1,v_2,J_0,J_1,J_2$ are not themselves 
independent functions of $u_i$: they are connected by the relation 
$$ 
I_1 =  v_1 -\frac{1}{v_0}-\frac{1}{v_{1}}-\frac{1}{v_{2}} =  J_0+J_{1}+J_{2}, 
$$ 
which implies that the first integrals satisfy 
$ I_1^2 =T_0 + 2T_1$.    
Subject to the latter relation, the Liouville integrability of the 
map (\ref{E:N=2}) 
follows by taking any three independent first integrals in involution ($I_0,I_1,T_0$, for instance). 
The existence of an additional independent first integral, namely 
$S_0=J_0J_{1}J_{2}$ 
(another symmetric function of $J_0,J_{1},J_{2}$), means that the map is superintegrable; 
but this integral does not Poisson commute with $T_0$ or $T_1$.
} 
\end{example}
\begin{example}\label{8d}
{\em 
For 
$d=4$,  equation~\eqref{E:2copies} gives the 
8-dimensional map 
\beq 
\label{E:N=3}
\left(u_{-4},u_{-3},\ldots, u_{3}\right)\mapsto \Big(u_{-3},\ldots, u_{3}, G\left(u_{-4},u_{-3},\ldots, u_{3}\right)\Big),
\eeq 
where
$
G\left(u_{-4},u_{-3},\ldots, u_{3}\right)=u_0-\Big(u_{1}+u_{-1}-u_{3}-u_{-3}+1/(u_{-4}-u_0)\Big)^{-1} 
$. 
This map is symplectic with respect to the nondegenerate Poisson structure given in Theorem~\ref{ubracket}.
There are three independent integrals  that come from the KdV map, 
given by 
$$ 
I_0=\left(\frac{1}{v_0v_1}-1\right)\, \left(v_1-\frac{1}{v_2}\right)\, \left(v_2-\frac{1}{v_3}\right), \qquad 
I_1 = {\cal C}_1 {\cal C}_2 - 1, \qquad 
{\cal K} =  {\cal C}_1 +{\cal C}_2, 
$$ 
which are all 
written in terms of $v_{i}=u_{i-4}-u_{i}$ for $i=0,1,2,3$, using 
$$ 
{\cal C}_1 = v_1-\frac{1}{v_0}-\frac{1}{v_2}, \qquad {\cal C}_2 = v_2-\frac{1}{v_1}-\frac{1}{v_3} . 
$$ 
The first two of these integrals ($I_0$ and $I_1$) are coefficients of  
the spectral curve in Example \ref{d4spec} (with $\nu =1$), 
while the third is not. 
As well as these, there are three independent cyclically symmetric quadratic 
functions of the 4-integrals 
$J_{i}=u_{i-3}-u_{i-1}- (u_{i-4}-u_{i})^{-1}$,  
$i=0,1,2,3$,  
namely 
$$ 
T_0 = J_0^2 +J_1^2 +J_2^2, \quad T_1 = J_0J_1+J_1J_2+J_2J_3+J_3J_0,  
\quad T_2 = 2(J_0J_2+J_1J_3),  
$$ 
which are also first integrals of (\ref{E:N=3}). There are two relations between the quantities 
$v_0,\ldots , v_3$ and $J_0,\ldots ,J_3$, as can seen by noting the identities  
$
{\cal C}_1 = J_0 + J_2
$, 
$
{\cal C}_2 = J_1 + J_3
$; 
thus the aforementioned first integrals are related by 
$ 
{\cal K}^2 =T_0+2T_1 +T_2
$, 
$
I_1+1 = T_1
$.  
Hence complete integrability of the map (\ref{E:N=3}) follows from 
the existence of 4 independent integrals in involution, i.e. $I_0,I_1,{\cal K}, T_0$. 
The map is also superintegrable, due to the presence of a fifth independent 
first integral, given by another symmetric function of the $J_i$, that is   
$
S_0=J_0J_1J_2J_3
$. 
}
\end{example} 
As we now briefly explain, 
the general case follows the pattern of one of the preceding two examples 
very closely, according to whether $d$ is odd or even. 

When  $d$ is odd, the spectral curve (\ref{spec}) has $(d+1)/2$ non-trivial 
coefficients, which are the quantities $I_r$ appearing in (\ref{trace}).  
There are also $(d+1)/2$ independent functions $T_s$, as in (\ref{ts}), but the identity 
$I_{(d-1)/2}= J_0+J_1+\ldots +J_{d-1}$ implies that these two sets of functions are  
related by 
$I_{(d-1)/2}^2 =\sum_{s=0}^{d-1}T_s$. 
Hence there are precisely $d$ independent functions in involution, as required for Theorem \ref{pkdvli}. 

In the case that $d$ is even, the non-trivial coefficients $I_r$ of the spectral curve 
are supplemented by the additional integral ${\cal K}={\cal C}_1+{\cal C}_2$, 
providing a total of  
$(d+2)/2$ independent functions, and there are the same number of independent 
functions of the form (\ref{ts}), but now the pair of identities 
$$
{\cal C}_1=\sum_{i\, \mathrm{even}, \, 0\leq i\leq d-2}J_i, \qquad  
{\cal C}_2=\sum_{i\, \mathrm{odd}, \, 1\leq i\leq d-1}J_i  
$$ 
together imply  that the first integrals satisfy the two relations 
${\cal K}^2 =\sum_{s=0}^{d-1}T_s$, 
$I_{(d-2)/2}+1=\sum_{s\,\mathrm{odd}, \, 1\leq s\leq d/2}T_s$, 
so once again there are $d$ independent integrals, as required. 

One can also construct extra first integrals $S_j$ for $j=0,\ldots ,\lfloor (d-1)/2\rfloor-1$, 
by taking additional  independent cyclically symmetric functions of the $J_i$. 
This means that the map (\ref{E:pKdVmap}) is superintegrable for all $d$. 

\subsection{Difference equations with periodic coefficients} 

To conclude this section, we look at equation~\eqref{E:gQRT} in a 
different way, and show how it is related to other difference equations 
with periodic coefficients. We begin by revisiting the previous two examples.   

\begin{example}\label{qrtp3} 
{\em 
For the map (\ref{E:N=2}), the introduction of the variables 
$x_{n}=u_{n}-u_{n+1}$ yields the following difference equation 
of second order: 
 \begin{equation}
 \label{E:gQRTN=2}
x_{n+2}+x_n=\frac{1}{x_{n+1}-J_n}-x_{n+1}, \qquad J_{n+3}=J_n. 
 \end{equation}
Iteration of the latter preserves the canonical symplectic 
form $\dd x_0\wedge \dd x_1$ in the $(x_0,x_1)$ plane. 
Apart from the coefficient $J_n$, 
there is another periodic quantity associated with this equation, 
namely the 3-integral $H_n$ which equals 
$$\begin{array}{ll} 
x_n^2x_{n+1}+x_{n+1}^2x_n 
&
+(J_{n+2}J_{n}-J_{n+1}J_{n}-1)x_{n} 
+ 
(J_{n+2}J_{n}-J_{n+1}J_{n+2}-1)x_{n+1} \\ 
& 
-J_{n}x_{n}^2-J_{n+2}x_{n+1}^2
+(J_{n+1}-J_{n}-J_{n+2})x_{n}x_{n+1}  
\end{array} 
$$ 
(with $H_{n+3}=H_n$). 
This is related to some of the first integrals for 
$d=3$ by the identity  
$H_n-J_{n+1}=I_0-S_0$,  
where $S_0=J_0J_{1}J_{2}$ as in Example \ref{6d}. 
} 
\end{example} 

\begin{example}\label{qrtp4} 
{\em 
By introducing $w_{n}=u_{n}-u_{n+2}$ in (\ref{E:N=3}), 
we obtain the second order
equation
\begin{equation}
\label{E:moreN=3}
w_{n+2}+w_n=\frac{1}{ w_{n+1}-J_n}, \qquad J_{n+4}=J_n.
\end{equation}
Each iteration of this difference equation 
preserves the canonical symplectic 
structure $\dd w_0\wedge\dd w_1$. 
Aside from the coefficient $J_n$, the equation (\ref{E:moreN=3}) 
has a 4-integral, i.e. $H_{n+4}=H_n$ where  $H_n$ is given by  
$$
\begin{array}{l} 
w_n^{2}w_{n+1}^{2} + \left( J_{{n+2}}-J_{{n}} \right) w_{n}^{2}w_{{n+1}}+ \left( J_{{n+1}}-J_{{n+3}} \right) w_{{n}}w_{n+1}^{2} \\ 
-J_{{n}}J_{{n+2}}w_n^{2}-J_{{n+1}}J_{{n+3}}w_{n+1}^{2}
+ \Big(  \left( J_{{n+3}}-J_{{n+1
}} \right)  \left( J_{{n}}-J_{{n+2}} \right) -1 \Big) w_{{n}}w_{{n+1
}} \\  
+\left(J_{{n}} J_{{n+1}}J_{{n+3}} -J_{{n+1}}J_{{n+2}}J_{{n+3}}-J_{{n+1}} 
\right) w_{{n+1}} 
\\ 
+ (
J_{{n}}J_{{n+2}}J_{{n+3}}-J_{{n}}J_{{n+1}}J_{{n+2}}-J_{{n+2}}) w_{{n}} 
.
\end{array} 
$$ 
One can check that $H_{n+1}-H_n=J_{n+2}(J_{n+3}-J_{n+1})$, and 
$H_n$ is related to the first integrals in Example \ref{8d} by the formula 
$
H_n-J_{n+1}J_{n+2}=I_0-S_0
$. 
} 
\end{example} 
\begin{remar} 
Both equations  (\ref{E:gQRTN=2}) and (\ref{E:moreN=3}) are of the type considered recently by Roberts 
\cite{nonqrt}: their orbits move periodically through a sequence of biquadratic curves, 
defined by the quantities $H_n$, and they reduce to symmetric QRT maps when $J_n=$constant.  
\end{remar}

In general, equation~\eqref{E:gQRT} can be rewritten in terms of the variables 
$x_{n}=u_{n-d}-u_{n-d+1}$, to obtain a difference equation of order $d-1$, that is 
\begin{equation}
\label{E:HigherQRT}
x_{n}+x_{n+1}+\ldots+x_{n+d-1}=\frac{1}{x_{n+1}+x_{n+2}+\ldots+x_{n+d-2}-J_n},  
\end{equation}
where $J_n$ is periodic with period $d$.
Note that both sides of the above equation are equal to $v_n$, the discrete KdV variable, 
as in (\ref{E:KdV_Reduction}).
This equation can be seen as a higher order analogue of the McMillan map, with periodic coefficients.
First integrals for equation~\eqref{E:HigherQRT} can be obtained from the integrals $I_r$ 
of the KdV map (\ref{E:kdvmap}) 
by rewriting all the variables $v_i$ (or $g_i$) in terms of $x_n$ and $J_n$.   
To be precise, one can verify that $g_{i}=x_{i-1}+x_{i}+J_{i}$ for $1\leq i\leq d-1$, 
and $v_0=-g_0^{-1}$ is given by the right hand side of~\eqref{E:HigherQRT} 
for $n=0$. In particular,  the explicit formula for $I_0$ is 
$$ 
I_0=\prod_{i=1}^{d-2}(x_{i-1}+x_{i}+J_{i})\left(\left(\sum_{i=1}^{d-2}x_i-J_0\right)
\left(\sum_{i=0}^{d-3}x_i-J_{d-1}\right)-1\right).
$$ 

When $d$ is even, one can reduce~\eqref{E:HigherQRT}  to a difference equation of order $d-2$, by setting 
$w_{n}=x_{n}+x_{n+1}$. Formulae for first integrals in terms of the variables $w_n$ follow directly from 
integrals given in terms of  $x_n$ and $J_n$.

\section{Conclusions} 

We have demonstrated that the 
$(d-1,-1)$-reduction of the discrete KdV equation 
(\ref{E:kdv}) is a completely integrable map in the Liouville 
sense. There are two different Poisson structures for 
this map: one was obtained by starting from the related 
double pKdV equation 
(\ref{double}) and its associated Lagrangian; the other 
arose by using tau-functions and a connection with cluster algebras. 
The appropriate reduction of the Lax pair (\ref{E:Laxpairs}) for discrete KdV, 
via the staircase method, was the key to finding explicit expressions 
for first integrals, and two ways were presented to prove that these 
are in involution. The corresponding reduction of the lattice equation 
(\ref{double}) was also seen to be completely integrable (and even 
superintegrable), with additional first integrals appearing 
due to the presence of the $d$-integral $J_n$.


An interesting feature of all these reductions is that, although they 
are autonomous difference equations, they have various  
difference equations with periodic coefficients associated with them, such as (\ref{dodd}), 
(\ref{deven}) and (\ref{E:HigherQRT}). 

There are several ways in which the results in this paper could be developed  
further. In particular, it would be interesting to understand the 
Poisson brackets for the reductions in terms of an appropriate r-matrix. 
It would also be instructive to make use of the bi-Hamiltonian structure to 
perform separation of variables, by the method in \cite{blaszaksov}, 
and this could be further used to obtain the exact solution of these 
difference equations in terms of theta functions associated with the 
hyperelliptic spectral curve (\ref{spec}). 
Finally, it would be worthwhile to extend the results here 
to $(s_1, s_2)$-reductions of discrete KdV, and to reductions of other integrable lattice equations 
that have not been considered before.

\begin{figure}\centering
\scalebox{0.5}[0.5]{
\includegraphics[angle=270]{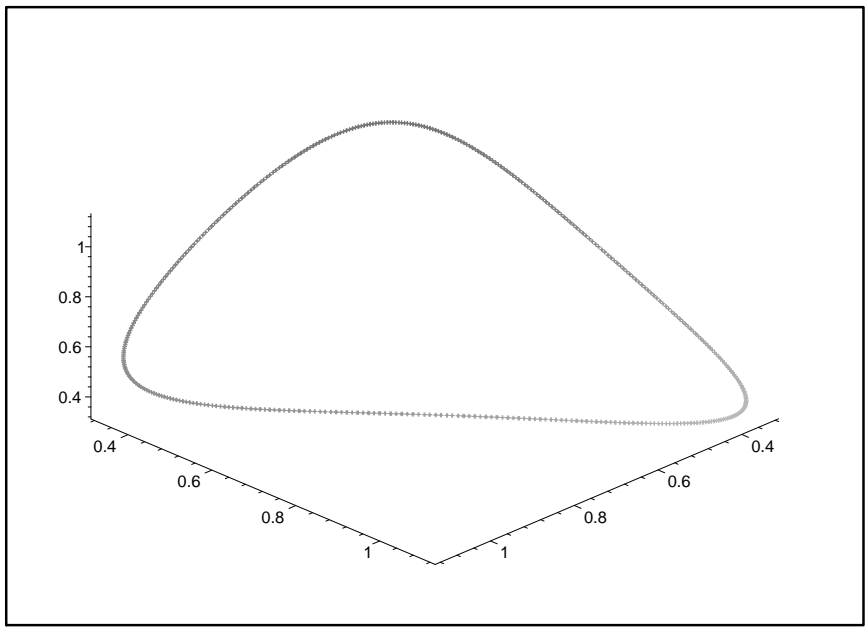}}
\caption{\small{The orbit of the map (3.2) for $d=3$ with initial data $(v_0,v_1,v_2)=\left(1,\frac{1}{2},\frac{2}{5}\right)$. }}
\label{pentangle} 
\end{figure}

\section*{Appendix} 

\subsection{Comment on the solution of the double pKdV equation} 

The double pKdV equation (1.3), which is equivalent to   
$$ 
\widehat{\widetilde{\widetilde{u}}}-\widehat{\widehat{\widetilde{u}}}- 
\left(\widehat{\widetilde{u}}-  \widehat{\widehat{\widetilde{\widetilde{u}}}} \right)^{-1} = 
\t{u} -\h{u}-\Big(u- \h{\t{u}}\Big)^{-1},  
$$ 
is more general than the pKdV equation (1.1), that is 
$$
(u-\widehat{\widetilde{u}})(\widetilde{u}-\widehat{u})=1. 
$$ 
Here we explain how it is possible to obtain any solution of (1.3) 
from a solution of pKdV together with the solution of a linear equation. 

To see this, let $a$ be a solution of the linear partial difference equation 
$$ 
\widehat{\widetilde{a}}=a.  
$$ 
(Note that, from (2.4),  this is the same as the linear equation satisfied by the quantity J.) 
Then it can be verified by direct calculation that  the combination 
$$ 
u=U+a 
$$ satisfies (1.3), where $U$ is any solution of (1.1). 
Conversely, let $u$ be any solution of (1.3), and define the quantity J in terms of $u$ 
according to the right hand formula in (2.4). Now let $a$ be a solution of the pair of linear equations 
$$ 
\t{a} -\underset{\sim}{a}=\mathrm{J}, \qquad   \widehat{\widetilde{a}}=a.
$$ 
It follows from (1.3) that the latter two equations are compatible, and the 
combination 
$$ 
U=u-a
$$ satisfies the pKdV equation (1.1). Hence any solution of (1.3) can be written in the form 
$u=U+a$. 

At the level of the $(d-1,-1)$-reduction of (1.3), this means that all solutions 
of the ordinary difference equation (2.7) can be written in the form 
$$u_n=U_n+a_n,$$ 
where $a_n$ is periodic with period $d$, as well 
as being a solution of the second order linear equation 
$$ 
a_{n+1} - a_{n-1} = J_n, 
$$ 
and $U_n$ is a 
solution of the  $(d-1,-1)$-reduction of (1.1) (as studied in reference [34]). 
It is not clear how this connection between the two reduced equations 
should be interpreted from the point of view of Liouville integrability.

\subsection{An example of an orbit for $d=3$} 

Consider the choice of initial values $\tau_0= \tau_1= \tau_2= \tau_3= 1$ for the Somos-4 recurrence 
with 2-periodic coefficients, as in Example 9, that is 
$$ 
\tau_{n+4}\, \tau_n =\alpha_n \, \tau_{n+3}\, \tau_{n+1} +  \tau_{n+2}^2, \qquad 
\mathrm{with}  
\qquad 
\alpha_{n+2}=\alpha_n, 
$$
together with the choice of parameters $\alpha_0=1$, $\alpha_1=2$.  
This corresponds to taking the initial values 
$v_0=1$, $v_1=1/2$, $v_2=2/5$ in the map (3.2) with $d=3$, which is 
given by 
$$ 
\varphi: \quad   
(v_0,v_{1},v_{2})\mapsto\left(v_{1},v_{2},\frac{v_0}{1+v_{2}v_0-v_{1}v_0}\right).
$$
In Figure 1 
we plot the orbit of the map $\varphi$ that is generated by this initial data.  


\section*{Acknowledgment} 

This work was supported by the Australian Research Council. 
DTT visited the University of Kent in 2011 and 2012, 
and is grateful  for the support of  an 
Edgar Smith Scholarship which 
funded her travel. %
ANWH would like to thank the organisers of the Nonlinear Dynamical Systems workshop for 
supporting his trip to La Trobe University,  
Melbourne in September-October 2012.

\end{document}